\documentclass[onecolumn,floats,floatfix,showpacs,nofootinbib,11pt]{revtex4-1}
\usepackage{graphicx}
\usepackage{dcolumn}
\usepackage{bm}
\usepackage{graphics}
\usepackage{slashed}
\usepackage{amssymb}
\usepackage{natbib}
\usepackage{amsmath}
\usepackage{url}
\usepackage[usenames,dvipsnames,svgnames,table]{xcolor}
\usepackage{amsfonts}
\usepackage{subfig}
\usepackage{amsthm}
\usepackage{hyperref}
\usepackage[retainorgcmds]{IEEEtrantools}

\newcommand{\bea}{\begin{eqnarray}}
\newcommand{\ena}{\end{eqnarray}}
\newcommand{\bean}{\begin{eqnarray*}}
\newcommand{\enan}{\end{eqnarray*}}

\newtheorem{theorem}{Theorem}
\newtheorem{lemma}{Lemma}
\newtheorem{proposition}{Proposition}
\newtheorem{remark}{Remark}
\newtheorem{definition}{Definition}

\begin{document}

\title{Nonlinear electrodynamics as a symmetric hyperbolic system}
\author{Fernando Abalos${}^{1}$}
\email{jfera18@gmail.com}
\author{Federico Carrasco${}^{1}$}
\email{fedecarrasco@gmail.com}
\author{\'Erico Goulart${}^{2}$}
\email{egoulart2@gmail.com}
\author{Oscar Reula${}^{1}$}
\email{reula@famaf.unc.edu.ar}

\affiliation{${}^{1}$ Facultad de Matem\'atica, Astronom\'\i{}a y F\'\i{}sica, Universidad Nacional de C\'ordoba and IFEG-CONICET, Ciudad Universitaria, X5016LAE C\'ordoba, Argentina }
\affiliation{${}^{2}$ CAPES Foundation, Ministry of Education, Brasilia, Distrito Federal, 70.040-020, Brazil}

\begin{abstract}

Nonlinear theories generalizing Maxwell's electromagnetism and arising from a Lagrangian formalism have dispersion relations in which propagation planes factor into null planes corresponding to two effective metrics which depend on the point-wise values of the electromagnetic field. 
These effective Lorentzian metrics share the null (generically two) directions of the electromagnetic field. 
We show that, the theory is symmetric hyperbolic if and only if the cones these metrics give rise to have a non-empty intersection. Namely that there exist families of symmetrizers in the sense of Geroch \cite{Geroch} which are positive definite for all covectors in the interior of the cones intersection. Thus, for these theories, the initial value problem is well-posed. We illustrate the power of this approach with several nonlinear models of physical interest such as Born-Infeld, Gauss-Bonnet and Euler-Heisenberg.

\end{abstract}


\maketitle



\section{Introduction}

Nonlinear electrodynamics (NLED) is relevant in several areas of physics. In QED, the polarization of the vacuum leads naturally to nonlinear effects (such as the light-light scattering) which are effectively described by Euler-Heisenberg's Lagrangian \cite{Eul}, \cite{Heis}, \cite{Karp}, \cite{Schwin}, \cite{Blau}, \cite{Soldati} (see also \cite{Dunne} for a pedagogical review). In some dielectrics and crystals, the interaction between the molecules and external electromagnetic fields can be described by an effective nonlinear theory, which is typically observed at very high light intensities such as those provided by pulsed lasers \cite{Shen1984}, \cite{Born1999}. Possible consequences of NLED have been explored also in cosmology and astrophysics. In particular, it is believed that nonlinearities may play important roles in the description of the dark sector of the universe \cite{breton2010nonlinear}, \cite{de2002nonlinear}, \cite{novello2004nonlinear}, \cite{novello2007cosmological}, \cite{garcia2014no}, \cite{montiel2014parameter}, in the avoidance of singularities (when coupled to Einstein's equations) or in the physics of charged black holes, \cite{plebanski1984type}, \cite{breton2003born},\cite{chemissany2008thermodynamics}, \cite{gunasekaran2012extended}, \cite{ayon1998regular}, \cite{breton2015thermodynamical}. These models have the advantage of using only electromagnetic fields, without invoking yet unobservable scalars or more speculative ideas related to higher-dimensions and brane worlds. Finally, Born-Infeld nonlinear model \cite{Born_1, Born_2, Born_3} has mathematical connections to string theory, for its Lagrangian appears in relation with the gauge fields on a D-brane (see, for instance, \cite{Gib}). In general, NLED theories attracts attention because they offer insights into light propagation in the experimental and theoretical studies of relativity.


A key feature to elucidate about the partial differential equations (PDEs) governing NLED is whether they pose an initial value formulation. Well-posedness is at the roots of physics, for it amounts to the predictability power of the theory, asserting that solutions exist, are unique and depend continuously on the initial data. The mathematical theory dealing with the initial value formulation is well developed and for this case it amounts to check whether the first order system of quasilinear PDEs is hyperbolic and what are the maximal propagation speeds. Roughly, hyperbolicity is an algebraic property of the principal symbol (the differential operator consisting of the highest derivative order terms of the PDE), which is essential to prove local well-posedness of the non-characteristic Cauchy problem. In theories like NLED the analysis of hyperbolicity needs, however, a careful manipulation. Because constraints are present, the evolution equations are not uniquely defined for one can add to any one of them 
constraint terms and obtain another equivalent evolution system. The distinction of an evolution system over the others generally implies a choice of $3+1$ decomposition of space-time and, as such, breaks covariance.

In this paper we shall treat the equivalent evolution systems in an equal footing so as to assert the desired hyperbolicity properties keeping covariance as much as we can. The tools to deal with hyperbolicity in this way has been provided by Geroch in \cite{Geroch}. However, as there are many different notions of hyperbolicity available in the literature (see, e.g. \cite{courant1966methods}, \cite{lax1}, \cite{lars1997lectures}), we now recall some previous results concerning the evolutionary aspects of NLED. In \cite{brenier2004hydrodynamic}, for instance, Brenier uses the energy density and the Poynting vector as additional unknowns to augment the original $6\times 6$ Born-Infeld system to a system of $10\times 10$ hyperbolic conservation laws (for which well known results can be directly applied \cite{lax1973hyperbolic}, \cite{majda2012compressible}, \cite{serre1999systems}). A similar analysis using convex entropies (actually the energy) may be found in \cite{Serr}, where Serres extends a method 
designed by C. 
Dafermos  \cite{dafermos2010hyperbolic} (see also \cite{demoulini2001variational}) to the class of models described by Coleman \& Dill \cite{coleman1971thermodynamic}. As a result, he shows that the polyconvexity of the energy density implies the local well-posedness of the Cauchy problem within smooth functions of class $H^{s}$ with $s>1+d/2$. Another approach using a $3+1$ splitting was provided by V. Perlick in \cite{Volker}, assuming that the constitutive equations $\mathcal{F}_{A}(H,F)=0$ can be solved for $\vec{E}$ and $\vec{H}$. Several results that are relevant for the question of whether the evolution equations are hyperbolic, strongly hyperbolic or symmetric hyperbolic are investigated in details. It is worth mentioning also the global results by J. Speck \cite{speck2012nonlinear}. Using ideas presented by Christodoulou and Klainerman \cite{christodoulou2000action}, \cite{christodoulou1990asymptotic} he establishes the existence of small-data global solutions to the Born-Infeld system on the 
Minkowski space background in 1 + 3 dimensions.
Roughly, he concludes that if the initial data for the Born-Infeld equations are sufficiently small as measured by a weighted Sobolev norm, then these data launch a unique classical
solution to the equations existing in \textit{all} of Minkowski space. Furthermore, he shows that these solutions decay at exactly the same rates as solutions to the linear Maxwell system.

Here, we adopt a different strategy. Working within Geroch's geometrical formalism we find the most general hyperbolizations NLED theories admit and show that they're parametrized by an auxiliary vector field $t^{q}(x)$ (as is also the case for linear electrodynamics). This construction allows us to find necessary and sufficient conditions theories and fields must satisfy in order to have a well-posed initial value formulation. It happens, such condition translates very nicely into geometrical terms: the system is symmetric hyperbolic, if and only if, the two cones arising from the dispersion relations (or, conversely, the characteristic surfaces) of the given NLED theory have a non-empty intersection. 
%
%
This constitutes one of the main results of the present article. Our construction also allows to find the causality cone (the maximal propagation speeds) each hyperbolization have. These cones could be, in principle, different from the physical cones as defined by the dispersion relations, for the latter ones use all equations, including the constraints, to assert the propagation velocities 
of plane waves, while the hyperbolization cones use only the corresponding evolution equations. We find that, nevertheless, they coincide. 

\textbf{A warning on terminology}: In the present context, cones appear in several related disguises: First, we have the familiar cones arising in Lorentzian geometry, namely those arising from the set of all time-like vectors $\{ v \in T_{p}\textbf{M} |  \text{  } g_{ab}v^a v^b > 0\}$, which splits into two disjoint sets, the ``future" and ``past" propagation cones. We shall reffer to them as the cones of a given metric; 
Second, now allowing for more general symmetric-hyperbolic systems, the cones that appear as the set of co-vectors which make positive definite a certain symmetric hyperbolizer, $C^*_H:=\{n_a \in T^{*}_{p}\textbf{M} |  \text{  } H_{\alpha \beta}^a n_a > 0\}$. These cones are clearly open and convex, for if $n_a$ and $n'_a$ belong to $C^*_H$ so does $\lambda_1 n_a + \lambda_2 n'_a$ for all positive $\lambda_1$, $\lambda_2$, since the sum of positive bi-linear forms  gives another  positive bi-linear form. They represent planes on the tangent space. Once a plane is found within $C^*_H$, the others are found tilting it until $H_{\alpha \beta}^a n_a$ gets a kernel.
Each one of these planes represents a plane wave perturbation solution to the underlying equations out of which the symmetrizer was built upon; 
Third, we can construct the co-cones, duals to the previous ones, that is, given a cone, $C \in V$ we can define a cone in the dual space, $C':= \{\sigma \in V' | \text{  } \sigma(v) \geq 0 \text{ ,  }\forall v \in C\}$, this is clearly also a convex cone. For the cones defined from hyperbolizers these cones are called the \textsl{propagation cones}, for they determine the directions along which perturbations propagate. 
Notice that the co-cones of a given metric cone are precisely those covectors obtained by lowering  with the metric the indices of the all the vectors forming the cones, so the distinction is lost.
However, since in this work there appear several Lorentzian metrics it is best to keep cones and co-cones as separate geometrical entities. \\ \\
This article is organized as follows: In Section II, we introduce NLED equations and point out some of their basic structural properties. 
We then review some important results concerning characteristic surfaces and deepen on the geometrical aspects these theories exhibit; 
In section III, an intrinsic geometrical formulation of PDEs due to Geroch is introduced and the notion of symmetric hyperbolicity is presented. 
We then investigate the algebraic core of the NLED equations and find the most general hyperbolizations these theories allow. Our main results are presented in this section in the form of two theorems, together with (we hope) comprehensive descriptions of the key ideas behind the proofs; In section IV, we provide the reader with the detailed steps to prove the theorems and thus while Section V deals with the constraints; Finally, Section VI explores our results for some particular realizations of the Lagrangian.



\section{General Remarks}

\subsection{Lagrangians and equations of motion}
Write $(\textbf{M},\ g)$ for a (1+3)-dimensional space-time, with signature convention $(+,-,-,-)$, and $F_{ab}$ for the electromagnetic 2-form. Let
\begin{equation}
F := F^{ab}F_{ab}=2(H^2-E^2)\quad\quad\quad G := F^{ab}\stackrel{\ast}{F_{ab}}= 4\vec{E}.\vec{H}
\end{equation}
denote the field invariants, where $\stackrel{\ast}{F_{ab}}=\frac{1}{2} \eta_{abcd} F^{cd}$ stands for the dual, $\eta_{abcd}$ the completely antisymmetric Levi-Civita tensor and $(\vec{E},\vec{H})$ the electromagnetic fields. We consider nonlinear models of electrodynamics in vacuum provided by the action,
\begin{equation}\label{action}
S=\int\mathcal{L}(F,G)\sqrt{-g}\ d^{4}x,
\end{equation}
where the Lagrangian density $\mathcal{L}(F,G)$ is an arbitrary smooth function of the invariants and $g:={det}(g_{ab})$. The first-order, quasi-linear,  equations read as
\begin{equation}\label{ql}
\nabla_{a}\left(\mathcal{L}_{F}F^{ab}+\mathcal{L}_{G}\stackrel{\ast}{F^{ab}}\right)=0,\quad\quad\quad\nabla_{[a}{F_{bc]}}=0,
\end{equation}
with $\nabla$ denoting covariant derivative and $\mathcal{L}_{X}:=\frac{\partial \mathcal{L}}{\partial X}$, for conciseness. Here, the l.h.s system is obtained via the variational principle while the r.h.s system is assumed from the very beginning, thus guaranteeing the existence of a four-potential such that $F_{ab}=\nabla_{[a}A_{b]}$. We define also, for future convenience, the quantities:
\begin{equation}
\xi_{1}:= 2\mathcal{L}_{FF}/\mathcal{L}_{F},\quad\quad\xi_{2}:= 2\mathcal{L}_{FG}/\mathcal{L}_{F},\quad\quad\xi_{3}:= 2\mathcal{L}_{GG}/\mathcal{L}_{F}.
\end{equation} 

\begin{remark}
In Maxwell's electrodynamics, there exists a particular gauge which considerably simplifies the second order equations for $A_{a}(x)$. Unfortunately, this is not the case for more general nonlinear theories and it will be convenient to focus our analysis on the first order equations \eqref{ql} only. 
Also, currents $j_{a}(x)$ are irrelevant for the question of whether the initial value problem is well-posed. Therefore, they'll play no role in our further discussion.
\end{remark}

 
\subsection{Dispersion relations, characteristics and effective metrics}
 
It is well known that high-frequency perturbations about a smooth background solution of Eqs. \eqref{ql} are controlled by two \textit{effective metrics}\footnote{See \cite{holes2002m} and \cite{barcelo2005analogue} for a wealth of details in the context of analogue models of gravity.}. The core of this result was first presented by Boillat \cite{Boi_1} and Plebanski \cite{Pleb} in the early 70's using Hadamard's method of discontinuities \cite{hadamard1903lecons}. More recently, Obukhov and Rubilar \cite{Obu} noticed that quasi-linear PDE's of the form \eqref{ql} are particular instances of electrodynamics inside media described by general nonlinear constitutive laws. As a consequence, they showed that if $\Sigma$ is a characteristic hypersurface described by $f(x^{a})=const$, the \textit{wave normals} $k_{m}:=\partial_{m}f$ are given by the vanishing sets of a fourth-order multivariate polynomial in the cotangent bundle $T^{*}\textbf{M}$
\begin{equation}\label{multiv}
\mathcal{P}^{*}(x,k):=\mathfrak{G}^{abcd}(x)\ k_{a}k_{b}k_{c}k_{d}=0.
\end{equation}
Here, $\mathfrak{G}^{abcd}(x)$ is a completely symmetric quantity (35 independent components) depending implicitly on the background solution. Physically, (\ref{multiv}) plays the role of a dispersion relation for the linearized waves and gives rise to some sort of \textit{covariant Fresnel equation} \cite{obukhov2000wave}, \cite{rubilar2002generally}.

For the class of nonlinear Lagrangian models provided by (\ref{action}) a remarkable property holds due to algebraic conditions: the multivariate polynomial (\ref{multiv}) always reduces to the simpler form
\begin{equation} \label{P}
\mathcal{P}^{*}\left(x,k\right) =ak^{4}+Qk^{2}l^{2}+Rl^{4}
\end{equation}
with $k^{2}=g^{ab}k_{a}k_{b}$, $l^{2}=F_{~c}^{a}F^{bc}k_{a}k_{b}$, for conciseness, and 
\begin{eqnarray}
a:&=& \left(1+\xi_{2}G-\xi_{3}F-RG^{2}/16\right),\\
Q:&=& 2\left(\xi_{1}+\xi_{3}-RF/4\right),\\
R:&=&4\left(\xi_{1}\xi_{3}-\xi_{2}^{2}\right).
\end{eqnarray}
A closer inspection of (\ref{P}) reveals that the quartic Fresnel surface of the wave normals factorizes to the product of two second order surfaces, given in terms of the quadratic forms 
\begin{equation}\label{quadratic}
\mathfrak{g}^{ab}_{1}(x)k_{a}k_{b}=0,\quad\quad\quad\quad \mathfrak{g}^{ab}_{2}(x)k_{a}k_{b}=0,
\end{equation}
with the \textit{reciprocal effective metrics} given by
\begin{eqnarray}
\mathfrak{g}_{1}^{ab}:&=&ag^{ab}+b_{1}F^{a}_{\phantom a c}F^{bc},\\
\mathfrak{g}_{2}^{ab}:&=&g^{ab}+(b_{2}/a)F^{a}_{\phantom a c}F^{bc}, 
\label{effective_met}
\end{eqnarray}
according to the following definitions
\footnote{If $a \to 0$ but $b_1 \neq 0$ the left hand side of (\ref{effective_met}) is still finite since $b_2/a = R/b_1$.
When $a \to 0$ and $b_1 \to 0$ simultaneously we can, using the conformal freedom, redefine the metrics as $\mathfrak{\tilde{g}}_{1}^{ab}:=(1/b_{1})\mathfrak{g}_{1}^{ab},\  \mathfrak{\tilde{g}}_{2}^{ab}:=b_{1}\tilde{g}_{2}^{ab}$. The results we shall obtain do not depend on this reparametrization.
\label{foot}}

\begin{equation}
b_{1}:=\frac{Q+\sqrt{\Delta}}{2},\quad\quad
b_{2}:=\frac{Q-\sqrt{\Delta}}{2},\quad\quad\Delta:=Q^{2}-4aR.
\end{equation}
As stressed by Boillat \cite{Boi_1}, the roots always exist since the discriminant is actually a sum of squares, i.e. $\Delta=4(N_{1}^{2}+N_{2}^{2})$ with
\begin{equation}
N_{1}:=(\xi_{1}-\xi_{3})-RF/4,\quad\quad\quad N_{2}:=2\xi_{2}-RG/4.
\end{equation}

Obukhov and Rubilar proceed by showing us that $\mathfrak{g}^{ab}_{1}$ and $\mathfrak{g}^{ab}_{2}$ are Lorentzian whenever the background spacetime metric is Lorentzian. Their results are in qualitative agreement with previous results by Novello \textit{et al} \cite{novello2000geometrical}, \cite{de2000light} wherein the rays spanning the characteristic surfaces are often described in terms of \textit{effective null geodesics}. Therefore, in what follows we shall assume that the effective metrics are always Lorentzian\footnote{A degenerate metric does not give rise to a cone according to our definition. Therefore, we analyze them separately in appendix B, where we conclude those systems are not symmetric hyperbolic.}. For the covariant components of the latter we write $\mathfrak{g}^{1}_{ab}$ and $\mathfrak{g}^{2}_{ab}$. They're such that $\mathfrak{g}_{1}^{ac}\mathfrak{g}_{cb}^{1}=\delta^{a}_{\phantom ab}$ and $\mathfrak{g}_{2}^{ac}\mathfrak{g}_{cb}^{2}=\delta^{a}_{\phantom ab}$. 





\subsection{Geometrical Structure: the cones}

The effective metrics are defined up to a conformal transformation. However, there are on them an intrinsically geometric property which does not depend on conformal redefinitions, namely their cones \footnote{By a cone we mean the interior of a proper cone. 
Recall that a proper cone $C$ is a subset of a vector space $V$ such that: $(\alpha \mathbf{u} + \beta \mathbf{v}) \in V \;\;\; 
\forall \text{  } \mathbf{u},\mathbf{v} \in V  \;\;\; \forall \text{  } \alpha, \beta >0$ and $\bar{C} \cap -\bar{C} = \emptyset$.}.
We will denote by $C_{\mathfrak{g}_{i}}$,  the set of all timelike (future-directed)~\footnote{For symmetric hyperbolic systems once a symmetrizer is given we will adopt for all effective metrics the convention that the future cones are those with a non-vanishing intersection with the corresponding propagation cone. In the case the system is not symmetric hyperbolic we shall take into account all cones, two for each metric, and refer to all of them as cones.} vectors with respect to the metrics $\mathfrak{g}_{i}$, 
and $C^{*}_{\mathfrak{g}_{i}} := int \left( \{ n_a \in T^{*}_p \textbf{M} | \text{  } n_a v^a > 0 \text{ ,  }\forall v^a \in C_{\mathfrak{g}_{i}}\} \right)$
\footnote{We choose the interior as to make the dual cones open, and thus, put them in an equal footing with the cones  $C_{\mathfrak{g}_{i}}$.} 
the dual of these cones, or co-cones.
We study here some geometrical relations among these cones, since they'll play an important role in the description of hyperbolicity. In order to do that, we distinguish between two different cases, according to the nature of the electromagnetic field $F_{ab}$: i. non-degenerate fields ($F^{2}+G^{2}\neq 0$); ii. degenerate fields ($F=G=0$). As these cases have different geometric interpretations, we shall treat them separately (see also \cite{de2009classification}).


\subsubsection{Non-degenerate $F_{ab}$}

A non-degenerate 2-form at a point $p$, has (at that point) two null eigenvectors $k^a$ and $l^a$ (see appendix A for further details).
The directions of these vectors are called the \textit{principal null directions (PND's)} \cite{penrose, pirani, wald2010general} and will play a key role in our description.
\begin{remark}
$k^a$ and $l^a$ are null with respect to the three metrics, that is, the background and the effective metrics. We shall say henceforth that the closures of all cones share these principal null directions. The plane spanned by $k^{a}$ and $l^{a}$ form the essential geometrical structure of the system.
\end{remark}
The PND's help us to construct a frame which simplifies calculations considerably.
In this frame, $\vec{E}\parallel\vec{H}$, $g_{ab}$ reduces to $\eta_{ab}$ and one obtains the quadratic forms 
\begin{equation}
\mathfrak{g}^{i}_{ab}t^{a}t^{b} :=\alpha _{i}^{-1}\left( t_{0}^{2}-t_{3}^{2}\right)
-\beta _{i}^{-1}\left( t_{1}^{2}+t_{2}^{2}\right) 
\qquad \mathfrak{g}_{i}^{ab}n_{a}n_{b}:=\alpha _{i}\left( n_{0}^{2}-n_{3}^{2}\right) -\beta
_{i}\left( n_{1}^{2}+n_{2}^{2}\right)
\label{gab_1}
\end{equation}
with $i=1,2$ and
\begin{equation}
\begin{array}{ccc}
\alpha _{1}:=a+b_{1}\frac{1}{4}\left( F-\sqrt{F^{2}+G^{2}}\right) &  & \beta
_{1}:=a+b_{1}\frac{1}{4}\left( F+\sqrt{F^{2}+G^{2}}\right) \\ 
\alpha _{2}:=1+\frac{b_{2}}{a}\frac{1}{4}\left( F-\sqrt{F^{2}+G^{2}}\right)
&  & \beta _{2}:=1+\frac{b_{2}}{a}\frac{1}{4}\left( F+\sqrt{F^{2}+G^{2}}%
\right)%
\end{array}
\label{alfa_beta1}
\end{equation}
Note that the quantities involved in \eqref{alfa_beta1} are covariantly defined, i.e. they are functions of the invariants $F$ and $G$.

\begin{figure}[!ht]
    \subfloat[$\Omega _{1}>0$ and $\Omega _{2}>0$ \label{cones:a}]{%
      \includegraphics[width=0.25\textwidth]{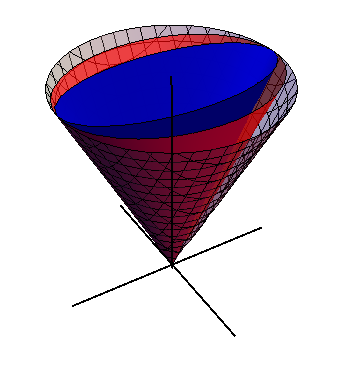}
    }
    \hfill
    \subfloat[$\Omega _{1}>0$ and $\Omega _{2}<0$ \label{cones:b}]{%
      \includegraphics[width=0.25\textwidth]{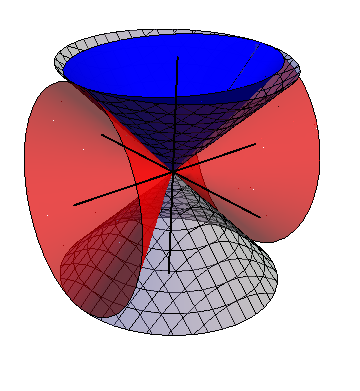}
    }
    \hfill
    \subfloat[$\Omega _{1}<0$ and $\Omega _{2}<0$ \label{cones:c}]{%
      \includegraphics[width=0.25\textwidth]{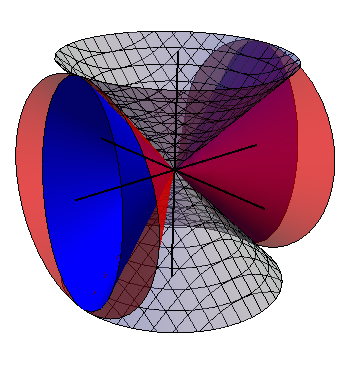}
    }
    \caption{Non-degenerate case. Possible configurations for the null surfaces of the metrics:\\ $\mathfrak{g}^{1}_{ab}$ (red); $\mathfrak{g}^{2}_{ab}$ (blue); background metric (meshed gray).}
    \label{fig:cones}
  \end{figure}
A closer inspection of \eqref{gab_1} reveals that the effective cones are characterized by the signs of the coefficients \eqref{alfa_beta1}. We introduce to this matter, two useful quantities
\begin{equation}
 \Omega_{i}:=\alpha _{i}\beta _{i} \qquad \text{(for } i=1,2 \text{  )}
\end{equation}
and present a table illustrating all possible sign combinations,
\begin{scriptsize}
\[
\begin{array}{ccccccc}
& \text{sign}\left( \alpha _{i}\right) & \text{sign}\left( \beta _{i}\right)
& \text{signature of }\mathfrak{g}_{i} & \text{timelike vector} & 
\mathfrak{g}^{i}_{ab}t^{a}t^{b} & sign(\Omega _{i})\\ 
i) & + & + & \left( +,-,-,-\right) & t^{a}=\left( 1,0,0,0\right) & \alpha_{1}^{-1}>0 & + \\ 
ii) & - & - & \left( -,+,+,+\right) & t^{a}=\left( 1,0,0,0\right) & \alpha_{1}^{-1}<0 & + \\ 
iii) & + & - & \left( +,+,+,-\right) & t^{a}=\left( 0,0,0,1\right) & -\alpha_{1}^{-1}<0 & - \\ 
iv) & - & + & \left( -,-,-,+\right) & t^{a}=\left( 0,0,0,1\right) & -\alpha_{1}^{-1}>0 & -%
\end{array}%
\] 
\end{scriptsize}
Note that the signatures of the effective metrics are not determined by their Lorentzian character. In particular, the norm of timelike vectors depends on the background field and does not have a preferred sign
\footnote{One could in principle redefine the metrics to fix both signatures in accordance with the convention chosen for the background metric,
but these redefinitions are going to depend generically on the background fields.}.
The cones, in the other hand, are clearly independent on how these signatures turn out; and the same will apply to all our future results, as we shall see later.
\begin{proposition}\label{prop_1}\text{  }\\
(i)  If $\Omega _{i}>0$, then $C_{\mathfrak{g}_{i}}\cap C_{\eta}\neq \emptyset $.  $\quad$(where $\eta^{ab}:=Diag\left( 1,-1,-1,-1\right) $).\\
(ii) If $\Omega _{i}<0$, then $C_{\mathfrak{g}_{i}}\cap C_{\sigma}\neq \emptyset $. $\quad$(where $\sigma^{ab}:=Diag\left( -1,-1,-1,1\right) $).
\end{proposition}

\noindent Therefore, there are only three qualitatively different configurations for the effective cones in $T_{p}\textbf{M}$: either the three cones intersect (Fig.\ref{cones:a}); 
the two effective cones intersect each other, but do not intersect the background cone (Fig.\ref{cones:c}); or they don't intersect (Fig.\ref{cones:b}). Moreover, when the effective cones do intersect each other, it turns out that is always possible to single one as being included inside the other. Thus, the following holds
\begin{proposition} \label{prop:inclusion} Whenever the effective cones intersect each other, i.e: $C_{\mathfrak{g}_{1}} \cap C_{\mathfrak{g}_{2}} \neq \emptyset$, then\\
(i) $\text{ }C^{*}_{\mathfrak{g}_{1}} \subseteq C^{*}_{\mathfrak{g}_{2}}  $\\  
(ii) $ C_{\mathfrak{g}_{2}} \subseteq C_{\mathfrak{g}_{1}}$
\end{proposition}
\begin{proof} 
We begin our proof with the important observation
\begin{lemma}
The following inequalities appear as algebraic consequences of the definitions, 
\begin{equation}
 \alpha _{1}\beta _{2} \leq 1 \leq \alpha _{2}\beta_{1}
 \label{lemma:alg}
\end{equation}
\end{lemma}
\noindent Indeed, starting from the full expressions,
 \begin{eqnarray}
\alpha _{1}\beta _{2} &=&1+\frac{1}{2}\left( \left( FN_{1}+GN_{2}\right) -%
\sqrt{N_{1}^{2}+N_{2}^{2}}\sqrt{F^{2}+G^{2}}\right)   \label{alfa1_beta2} \\
\alpha _{2}\beta _{1} &=&1+\frac{1}{2}\left( \left( FN_{1}+GN_{2}\right) +%
\sqrt{N_{1}^{2}+N_{2}^{2}}\sqrt{F^{2}+G^{2}}\right)   \nonumber
\end{eqnarray}
and the trivial inequality $\alpha _{2}\beta _{1} \geq \alpha_{1}\beta_{2}$ one obtains, after some simple manipulations,
\begin{equation}
-\left( 1-\alpha _{1}\beta _{2}\right) \left( 1-\alpha _{2}\beta
_{1}\right) = \frac{1}{4}\left( G N_{1}-F N_{2} \right)^{2} \geq 0. \label{ec_imp}
\end{equation}
\noindent The positivity of the left hand side, together with condition $\alpha _{2}\beta _{1} \geq \alpha_{1}\beta_{2}$, forces the inequalities of the lemma to hold. 

We now define two auxiliary quantities which capture the notion of ``how much the given cone (or co-cone) opens in any direction orthogonal to the PND's plane'' in this particular frame.
\begin{equation}
\label{gamma1} \gamma_i := \sqrt{\left\vert \alpha_i/\beta_i\right\vert },
\end{equation}
In spite of the fact that this is a coordinate-dependent notion it allows us to compare the cones and tell which of them is included in the other. 
The latter statement has a well defined geometrical meaning and directly extrapolates to any other frame.
The argument goes as follows. The assumption $C_{\mathfrak{g}_{1}} \cap C_{\mathfrak{g}_{2}} \neq \emptyset$ translates into
\begin{equation}
 sign(\Omega_1 ) = sign(\Omega_2 )  \qquad \text{or equivalently,} \qquad 0 < \Omega_1 \Omega_2 =\alpha _{1}\beta _{2}\alpha _{2}\beta _{1}
\end{equation}
which, together with \eqref{lemma:alg}, implies $\alpha_{1}\beta_{2}>0$ and, consequently, $\gamma _{1}\leq \gamma _{2}$.
Now, suppose $\Omega_1 > 0 $ and $\Omega_2 > 0 $ and consider a continuous transition from a time-like to a space-like covector (w.r.t. the effective metrics) parametrized by $x\in \lbrack 0,1]$ in the form
\begin{equation}
n_{a}\left( x\right) =\left( 1-x\right) \left( 1,0,0,0\right) +x\left(0,\cos \phi ,\sin \phi ,0\right)  \label{10_n}
\end{equation}%
with $\phi \in [0,2\pi)$. The idea is to find $x_1$ and $x_2$ such that the covector becomes null, i.e: $\mathfrak{g}_{i}^{ab}n_{a}( x_i ) n_{b}( x_i ) =0$ (i=1,2). The solutions are given by
\begin{equation}
 x_i =\frac{1}{1+\gamma_{i}^{-1}}
\end{equation}
which determines $x_1 \leq x_2 $. Conversely, assuming $\Omega_1 < 0 $ and $\Omega_2 < 0 $ we consider the family
$$n_{a}(x)=\left( 1-x\right) \left(0,0,0,1\right) +x\left( 0,\cos \phi ,\sin \phi ,0\right) $$
and get precisely the same answer i.e. $x_1 \leq x_2 $. This allows us to conclude that $\text{ }C^{*}_{\mathfrak{g}_{1}} \subseteq C^{*}_{\mathfrak{g}_{2}} $ as claimed in (i).  In order to prove (ii), we basically apply the same strategy. Instead of using the reciprocal effective metrics, we now look at the covariant objects $\mathfrak{g}^{i}_{ab}$ and the relevant quantities become $\gamma_{i}^{\ast}:= \gamma^{-1}_{i}$. One then concludes that our previous inequality is inverted i.e. $x_{1}\geq x_{2}$. The latter directly leads us to the desired result (ii).  

\end{proof}


\subsubsection{Degenerate $F_{ab}$}

For a degenerate 2-form, the two null directions collapse into a single one (see figure \ref{fig:deg}). Roughly, this means that, either the cones do not intersect each other, or one of them is included in the other\footnote{this do not exclude the possibility they coincide.}. Again, we refer the reader to appendix A for more details on frames. In particular, one can reduce the metrics and inverses to,
\begin{equation}
\begin{array}{c}
\mathfrak{g}^{i}_{ab}t^{a}t^{b} =\allowbreak \left(
t_{0}^2 - t_{3}^2 \right) + \varepsilon^{4}b_{i}\left( t_{0}-t_{3}\right)^2
-\left( t_{1}^{2}+t_{2}^{2}\right) \\ 
\mathfrak{g}_{i}^{ab}n_{a}n_{b} =\allowbreak \left(
n_{0}^2 - n_{3}^2 \right) -\varepsilon^{4}b_{i}\left( n_{0}+n_{3}\right)^2
-\left( n_{1}^{2}+n_{2}^{2}\right) 
\end{array}
\label{gab_deg}
\end{equation}
with $\varepsilon \neq 0$ a free parameter, a remaining freedom in our frame choice that we'll fix in a convenient way later on. Note when $b_i=0$ the effective metric reduce to background metric.

\begin{proposition} \label{prop:deg} For a degenerate $F_{ab}$, the effective cones always intersect, i.e: $C^{*}_{\mathfrak{g}_{1}} \cap C^{*}_{\mathfrak{g}_{2}} \neq \emptyset$, and with background metric $C^{*}_{\mathfrak{g}_{1}} \cap C^{*}_{\mathfrak{g}_{2}} \cap C^{*}_{\mathfrak{\eta}} \neq \emptyset$ . Furthermore,\\
(i) $\text{ }C^{*}_{\mathfrak{g}_{1}} \subseteq C^{*}_{\mathfrak{g}_{2}}  $\\  
(ii) $ C_{\mathfrak{g}_{2}} \subseteq C_{\mathfrak{g}_{1}}$
\end{proposition}
\begin{proof}
First, we want to see whether there is a non-empty intersection among the effective co-cones. The idea is to choose the parameter $\varepsilon $ ``sufficiently small'', so as to construct an explicit timelike covector w.r.t. both reciprocal effective metrics
\begin{lemma}\label{free-eps}
If $\varepsilon^4 < \min_{i=1,2} \left( |b_i |^{-1} \right)$, there exists a covector which is time-like with respect to both reciprocal effective metrics.
\end{lemma}
\begin{proof}
Let us consider, in the present frame, four linearly independent covectors:\\ 
$\tau_a = (1,0,0,0)$; $x_a = (0,1,0,0)$; $y_a = (0,0,1,0)$; $z_a = (0,0,0,1)$. One obtains,
\[
\begin{array}{ccc}
\mathfrak{g}_{i}^{ab}\tau_{a}\tau_{b}= \left( 1-\varepsilon^{4}b_{i}\right) >0 &  & 
\mathfrak{g}_{i}^{ab}x_{a}x_{b}=-1 \\ 
\mathfrak{g}_{i}^{ab}y_{a}y_{b}=-1
&  & \mathfrak{g}_{i}^{ab}z_{a}z_{b}=-\left( 1+\varepsilon ^{4}b_{i}\right) <0%
\end{array}%
\]
Therefore, there exists a common time-like covector $\tau_a$, and thus  $C^{*}_{\mathfrak{g}_{1}} \cap C^{*}_{\mathfrak{g}_{2}} \cap C^{*}_{\mathfrak{\eta}} \neq \emptyset$.
\end{proof}
To prove \textit{(i)}, we proceed in the same lines as in the non-degenerate case. Firstly, from the definitions of $b_1$ and $b_2$ there follows $b_2 \leq b_1$. We then consider a continuous transformation from a timelike to a space-like covector, parametrized by $x\in \lbrack 0,1]$ in the form: 
\begin{equation}
n_{a}\left( x\right) =\left( 1-x\right) \tau_a + x z_a 
\end{equation}%
with $\tau_a$ and $z_a$ as defined above. Again, we look for $x_1$ and $x_2$ within the range $\lbrack 0,1]$, for which the covector becomes null w. r. t. both metrics, i.e: $g_{i}^{ab}n_{a}( x_i ) n_{b}( x_i ) =0$ (for $ i=1,2$ ).
Then,
\begin{equation}
 x_i =\frac{1-\varepsilon^4 b_i}{2} 
\end{equation}
which implies $x_1 \leq x_2 $; Consequently,  $\text{ }C^{*}_{\mathfrak{g}_{1}} \subseteq C^{*}_{\mathfrak{g}_{2}}$
\footnote{
Notice that, in order to conclude this, it is enough to start from the common time-like covector and ``move'' in any direction,since we already know there is one (and just one!) common null direction among the effective metrics.
Thus, we argue it is not possible to find different results in different directions. 
As observed in the beginning of this section, once they have non-empty intersection, one of the cones must be included on the other. 
}. We remark that (ii) can be easily obtained in the same lines.

\end{proof}

\begin{figure}[!ht]
    \subfloat[$0 < b_{2}\leq b_{1}$  \label{:dummy}]{%
      \includegraphics[width=0.25\textwidth]{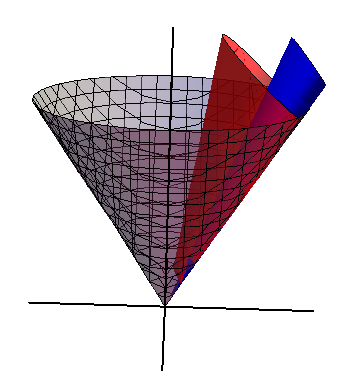}
    }
    \hfill
    \subfloat[$ b_{2}<0< b_{1}$ \label{subfig-2:dumm}]{%
      \includegraphics[width=0.25\textwidth]{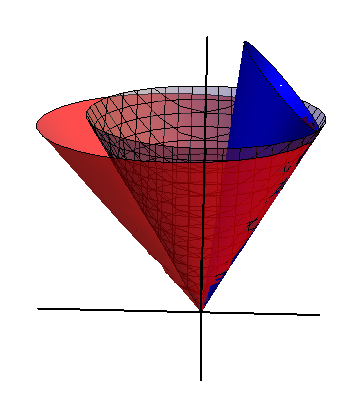}
    }
    \hfill
    \subfloat[$ b_{2}\leq b_{1}<0$ \label{subfig-2:dummy}]{%
      \includegraphics[width=0.3\textwidth]{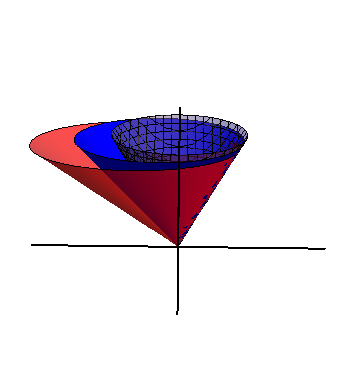}
    }
    \caption{Degenerate case. Possible configurations for the null surfaces of the metrics:\\ $\mathfrak{g}^{1}_{ab}$ (red); $\mathfrak{g}^{2}_{ab}$ (blue); background metric (meshed gray).}
    \label{fig:deg}
  \end{figure}

 
\section{Hyperbolizations}
In order to analyze the evolutionary properties of the system \eqref{ql} more closely, we shall recast it in the geometrical framework suggested by Geroch \cite{Geroch}:
\begin{equation}\label{KF}
K_{A \phantom a \alpha}^{\phantom a m}(x,\Phi)\partial_{m}\Phi^{\alpha}+J_{A}(x,\Phi)=0,
\end{equation}
where $K_{A \phantom a \alpha}^{\phantom a m}(x,\Phi)$ is called the principal part and $J_{A}(x,\Phi)$ stands for semi-linear contributions (whose explicit form is unnecessary for our discussion here). In this expression capital Latin indices, $A$, stand for the space of tensorial equations, lower Latin indices, $m$, for space-time indices and Greek indices, $\alpha$, for multi-tensorial unknowns. 

\begin{remark}
Typically, smooth solutions of \eqref{KF} are interpreted as cross-sections $\Phi^{\alpha}(x)$ over a smooth fibre bundle $\mathcal{B}$, with points $\kappa=(x^{a},\Phi^{\alpha})$ and we interpret the fibre over $x^{a}$ as the space of allowed physical states at $x^{a}$, i.e., as the space of possible field-values at that point.
\end{remark}

\begin{definition}
By a hyperbolization of \eqref{KF} over a sub-manifold $\textbf{S}\in\textbf{M}$, we mean a smooth symmetrizer $h^{A}_{\phantom a\alpha}$ such that:
\begin{enumerate}
\item{the field $h^{A}_{\phantom a\alpha}K_{A \phantom a \beta}^{\phantom a m}$ is symmetric in $\alpha,\ \beta$} in $\textbf{S}$;
\item{there exists a covector $n_{m}\in T^{*}\textbf{S}$ such that $h^{A}_{\phantom a\alpha}K_{A \phantom a \beta}^{\phantom a m}n_{m}$ is positive-definite.}
\end{enumerate}
If a system of first-order PDEs admits a symmetrizer satisfying the above conditions, we say that it is \textbf{symmetric hyperbolic}.
\end{definition}

\noindent Once a hyperbolizer exists, standard theorems apply and we know that given any smooth data in a hypesurface such that $n_a$ is normal to it a local solution for it would exists. Notice that since the set of co-vectors $n_a$ for which the hyperbolizer is positive is open, we can always choose them in a neighborhood of a point so that they are surface forming. We introduce also, for later convenience, the notion of physical propagation in this context. We denote by $C_{H}^{*}$ the collection of all covectors $\ n_{a}$ satisfying condition $(\textit{2.})$ above. 
Then, $C_{H}^{*}$ is a non-empty open convex cone.
\begin{definition}
The ``signal-propagation directions'' will be given by all tangent vectors $p^{a} $, such that $p^{a}n_{a}>0$, $\forall$ $n_{a}\in C_{H}^{*}$. The set of $p^a$ also forms a (non-empty) closed convex cone denoted by $C_{H}$, the ``dual cone'' of $C_{H}^{*}$.
\end{definition}
\begin{remark}
These cones could depend, in principle, on the hyperbolization selected. 
But it turns out that, for most physical examples, these cones are essentially independent of hyperbolization.
\end{remark}

For NLED, we  declare $\Phi^{\alpha}\rightarrow F^{bc}$, and a closer inspection of \eqref{ql} allows us to read-off the principal symbol
\begin{equation}\label{K}
K_{A \phantom a \alpha}^{\phantom a m}\rightarrow \left(-\left[g_{a\phantom a\phantom a bc}^{\phantom a m}+\ F_{a}^{\phantom a m}\big(\xi_{1}F_{bc}+\xi_{2}\stackrel{\ast}{F_{bc}}\big)+\stackrel{\ast} {F_{a}^{\phantom a m}}\big(\xi_{2}F_{bc}+\xi_{3}\stackrel{\ast}{F_{bc}}\big)\right],\ \frac{1}{2}\eta_{a\phantom a\ bc}^{\phantom a m}\right).
\end{equation}
with $g_{abcd}:= \frac{1}{2} (g_{ac}g_{bd}-g_{ad}g_{bc})$. In what follows we shall see that symmetric hyperbolicity holds for NLED under some basic assumptions. Basically, the latter are conditions on the Lagrangian, its derivatives and field strengths.

\subsection{Symmetrizer}
Our first task is to find a symmetrizer for \eqref{K}. In other words, we look for a $h^{A}_{\phantom a\alpha}$ such that $\delta\Phi{}^{\alpha}(h^{A}_{\phantom a\alpha}K_{A \phantom a \beta}^{\phantom a m})\delta\hat{\Phi}{}^{\beta}$ 
is symmetric in $\delta\Phi$ and $\delta\hat{\Phi}$. Making the identifications $\delta\Phi^{\alpha}\rightarrow X^{ab}$ and $\delta\hat{\Phi}{}^{\alpha}\rightarrow Y^{ab}$, with $X^{ab}$ and $Y^{ab}$ arbitrary anti-symmetric tensors, one has
\begin{equation}\label{eqn}
K_{A \phantom a \beta}^{\phantom a m}\delta\hat{\Phi}{}^{\beta}=\left(-\left[Y_{a}^{\phantom a m}+A_{\textbf{Y}}\ F_{a}^{\phantom a m}+B_{\textbf{Y}}\stackrel{\ast}{F_{a}^{\phantom a m}}\right]\ ,\ \stackrel{\ast}{Y_{a}^{\phantom a m}}\right),
\end{equation}
where
\begin{equation}
A_{\textbf{Y}}:= [\xi_{1}(\textbf{F}.\textbf{Y})+\xi_{2}(\stackrel{\ast}{\textbf{F}}.\textbf{Y})]\quad\quad\quad B_{\textbf{Y}}:= [\xi_{2}(\textbf{F}.\textbf{Y})+\xi_{3}(\stackrel{\ast}{\textbf{F}}.\textbf{Y})]
\end{equation}
and  $\textbf{X}.\textbf{Y}\equiv X^{ab}Y_{ab}$. Note that the above relations depend on the background field, the background metric and on the particular Lagrangian theory. It is natural to expect that $h^{A}_{\phantom a\alpha}$ depends also on these quantities. Also, when the theory is linear, i.e. $A_{\textbf{Y}}=B_{\textbf{Y}}=0$, we know that
\begin{equation}\label{max}
h^{A}_{\phantom a\alpha}\delta\Phi{}^{\alpha}=\left(X^{a}_{\phantom a q}\ ,\  -\stackrel{\ast}{X^{a}_{\phantom a q}}\right)t^{q}\end{equation}
where $t^{q}$ is an auxiliary smooth vector field. Therefore, the full nonlinear symmetrizer must: i) reduce to \eqref{max} in the linear case and ii) depend on at least one smooth vector field $t^{q}(x)$. We claim that it is given by the formula
\begin{equation}\label{sym}
h^{A}_{\phantom a\alpha}\delta\Phi{}^{\alpha}=\left(X^{a}_{\phantom a q}\ ,\  -\stackrel{\ast}{X^{a}_{\phantom a q}}-A_{\textbf{X}}\stackrel{\ast}{F^{a}_{\phantom a q}}+B_{\textbf{X}}F^{a}_{\phantom a q}\right)t^{q}.
\end{equation}
To show that it is indeed a symmetrizer, we first multiply \eqref{sym} by \eqref{eqn}. It follows,
\begin{eqnarray*}
\delta\Phi{}^{\alpha}(h^{A}_{\phantom a\alpha}K_{A \phantom a \beta}^{\phantom a m})\delta\hat{\Phi}{}^{\beta} &=&\Big(-X^{a}_{\phantom a q}Y_{a}^{\phantom a m}-A_{\textbf{Y}}X^{a}_{\phantom a q} \ F_{a}^{\phantom a m}-B_{\textbf{Y}}X^{a}_{\phantom a q}\stackrel{\ast}{F_{a}^{\phantom a m}}-\\
&&\quad -\stackrel{\ast}{X^{a}_{\phantom a q}}\stackrel{\ast}{Y_{a}^{\phantom a m}}-A_{\textbf{X}}\stackrel{\ast}{F^{a}_{\phantom a q}}\stackrel{\ast}{Y_{a}^{\phantom a m}}+B_{\textbf{X}}F^{a}_{\phantom a q}\stackrel{\ast}{Y_{a}^{\phantom a m}}\Big)t^{q}.
\end{eqnarray*}
Recalling that any pair of anti-symmetric tensors satisfies
\begin{equation}
\stackrel{\ast}{X^{aq}}\stackrel{\ast}{Y_{am}}=-\frac{1}{2}(\textbf{X}.\textbf{Y})\delta_{m}^{\phantom a q}+X_{am}Y^{aq},\quad\quad\quad\quad X^{aq}\stackrel{\ast}{Y_{am}}=\frac{1}{2}(\stackrel{\ast}{\textbf{X}}.\textbf{Y})\delta_{m}^{\phantom a q}-\stackrel{\ast}{X_{am}}Y^{aq},
\end{equation}
we obtain
\begin{equation}\label{exp}
\delta\Phi^{\alpha}(h^{A}_{\phantom a\alpha}K_{A \phantom a \beta}^{\phantom a m})\delta\hat{\Phi}{}^{\beta} =\big(M_{q}^{\phantom a m}+N_{q}^{\phantom a m}+L_{q}^{\phantom a m}\big)t^{q},
\end{equation}
with
\begin{eqnarray*}
M_{q}^{\phantom a m}&=&+(X^{\phantom a a}_{q}Y_{a}^{\phantom a m}+Y^{\phantom a a}_{q}X_{a}^{\phantom a m})+\frac{1}{2}(\textbf{X}.\textbf{Y})\delta_{q}^{\phantom a m},\\
N_{q}^{\phantom a m}&=&-(A_{\textbf{Y}}X^{a}_{\phantom a q}+A_{\textbf{X}}Y^{a}_{\phantom a q}) \ F_{a}^{\phantom a m}-(B_{\textbf{Y}}X^{a}_{\phantom a q}+B_{\textbf{X}}Y^{a}_{\phantom a q})\stackrel{\ast}{F_{a}^{\phantom a m}},\\
L_{q}^{\phantom a m}&=&+\{\xi_{1}(\textbf{F}.\textbf{X})(\textbf{F}.\textbf{Y})+\xi_{2}[(\stackrel{\ast}{\textbf{F}}.\textbf{X})(\textbf{F}.\textbf{Y})+(\stackrel{\ast}{\textbf{F}}.\textbf{Y})(\textbf{F}.\textbf{X})]+\xi_{3}(\stackrel{\ast}{\textbf{F}}.\textbf{X})(\stackrel{\ast}{\textbf{F}}.\textbf{Y})\}\delta_{q}^{\phantom a m}/2,
\end{eqnarray*}
which are symmetric quantities in X and Y. Note that this result is general and does not depend on the specific form of the Lagrangian. Note also that the symmetrizer splits into a main term, which coincides with Maxwell's symmetrizer and two nonlinear terms involving the background 2-form and its dual.
\subsection{Positive definiteness}

We now investigate when the symmetrizer constitutes a hyperbolization of the equations of motion. For the sake of conciseness, let us define the symmetric object
\begin{equation}
H_{\alpha\beta}(t,n)\equiv  h^{A}_{\phantom a\alpha}(t) K_{A\phantom a\beta}^{\phantom a m}n_{m}
\end{equation}
We emphasize here the linear dependence of this object with respect to both the vector $t^a$ (in the symmetrizer) and the covector $n_m$, a fact that will be important later on.
We shall see that the admissible range for the quantities $(t^{q},n_{m})$ has a nice geometrical interpretation in terms of the effective metrics.

According to our previous definition, the system will be symmetric hyperbolic if $H_{\alpha\beta}(t,n)$ constitutes a positive definite bi-linear form, i.e.
\begin{equation}\label{positivity-cond}
\Phi^{\alpha}H_{\alpha\beta}(t,n)\Phi^{\beta}>0
\end{equation}
for any non-zero $\Phi^{\alpha}$. 
Here, $\Phi^{\alpha}$ represents an arbitrary two-form, and as such can be thought as a vector in $\mathbb{R}^6$. Therefore, the family of symmetric maps $H_{\alpha\beta}(t,n): \mathbb{R}^6\rightarrow\mathbb{R}$. One can build a natural basis for this space by taking the six possible anti-symmetrized pairs of basis-elements in $T_{p}^{*}\textbf{M}$.

From \eqref{exp} we get 
\begin{equation}
H_{\alpha\beta}(t,n)=H_{\alpha\beta}(t,n)\big\vert_{M}+H_{\alpha\beta}(t,n)\big\vert_{N}+H_{\alpha\beta}(t,n)\big\vert_{L}
\end{equation}
A tedious but straightforward calculation yields
\begin{eqnarray*}
&&H_{\alpha\beta}(t, n)\big\vert_{M}=\Big(g_{abq[c}\delta^{m}_{\phantom a d]}+g_{cdq[a}\delta^{m}_{\phantom a b]}+g_{abcd}\delta^{m}_{\phantom a q}\Big)t^{q}n_{m},\\
&&H_{\alpha\beta}(t, n)\big\vert_{N}=\Big(\big(\xi_{1}F_{ab}+\xi_{2}\stackrel{\ast}{F_{ab}}\big)g_{q[c}F_{d]}^{\phantom a m}+\big(\xi_{1}F_{cd}+\xi_{2}\stackrel{\ast}{F_{cd}}\big)g_{q[a}F_{b]}^{\phantom a m} +\\
&&\quad\quad\quad\quad\quad\ +\ \ \big(\xi_{2}F_{ab}+\xi_{3}\stackrel{\ast}{F_{ab}}\big)g_{q[c}\stackrel{\ast}{F_{d]}^{\phantom a m}}+\big(\xi_{2}F_{cd}+\xi_{3}\stackrel{\ast}{F_{cd}}\big)g_{q[a}\stackrel{\ast}{F_{b]}^{\phantom a m}}\Big)t^{q}n_{m},\\
&&H_{\alpha\beta}(t, n)\big\vert_{L}=\frac{1}{2}\Big(\xi_{1}F_{ab}F_{cd}+\xi_{2}\big(\stackrel{\ast}{F_{ab}}F_{cd}+\stackrel{\ast}{F_{ab}}F_{cd}\big)+\xi_{3}\stackrel{\ast}{F_{ab}}\stackrel{\ast}{F_{cd}}\Big)(t.n).
\end{eqnarray*}

\vspace{0.2cm}

We now address the following: What are the conditions on the particular NLED theory (and fields) in order to have at least a pair $(t^a , n_a )$ satisfying the positivity requirement \eqref{positivity-cond}. 
Among the symmetric hyperbolic cases, what are the sets of admissible $t^a$ giving rise to a hyperbolization and what are the associated propagation cones?
The answers to these questions are provided by the following theorems\\

\begin{theorem}\label{theo_1}
The system is symmetric hyperbolic iff the  effective cones (co-cones) have a non-empty intersection, i.e. $C_{\mathfrak{g}_{1}} \cap C_{\mathfrak{g}_{2}} \neq \emptyset $ ($C^{*}_{\mathfrak{g}_{1}} \cap C^{*}_{\mathfrak{g}_{2}} \neq \emptyset $). This will be the case whenever\\ 
\begin{equation}\label{cod_sh}
\alpha _{1}\beta _{2} = 1+\frac{1}{2}\left( \left( FN_{1}+GN_{2}\right) -%
\sqrt{N_{1}^{2}+N_{2}^{2}}\sqrt{F^{2}+G^{2}}\right)> 0       
\end{equation}
\end{theorem}

\vspace{0.5cm}

\noindent The strategy behind our proof is as follows: write $H_{\alpha \beta}(t, n)$ in a preferred representation and build a mixed matrix $T^{\alpha}_{~\beta}$ (by means of an auxiliary Riemannian inner product $p_{\alpha \beta}$ of $\mathbb{R}^6$), such that:
\begin{equation}
H_{\alpha\beta }=p_{\alpha \gamma }T_{~\beta }^{\gamma }. \label{T}
\end{equation}
In particular, we choose  $p_{\alpha \beta }=\delta _{\alpha
\beta }$, with $\delta_{\alpha\beta}$ the six-dimensional Kronecker delta. We then look for the positiveness of all the eigenvalues of $T_{~~\gamma }^{\alpha }~$ which is essentialy equivalent to the positive definiteness of $H_{\alpha\beta }$\footnote{The reason to introduce $p_{\alpha \beta }$ is simply because it only makes sense to talk about eigenvalues when the operator acts from one vector space to itself. It might sound as a rather technical subtlety, but is not, since there is not unique (or natural) metric to raise the index here. There is an arbitrariness involved on the inner product and the resulting eigenvalues are not going to be covariantly definded in general; one is implicitly introducing a coordinate dependence on the election of a particular  $p_{\alpha \beta }$. Nevertheless, the positive character of the eigenvalues is in fact invariant and it will guarantee the positive definiteness of  $H_{\alpha\beta }$, as long as $p_{\alpha \beta }$ is an inner product.}. It turns out that the 
above requirement imposes a restriction on the allowed values of the background invariants 
(and on the theories, through $\xi_i$): such restriction is precisely  \eqref{cod_sh}, and it is equivalent to the statement that the effective cones do intersect each other. This result reveals an interesting geometrical aspect of NLED regarding hyperbolicity. Also, it provides a simple diagnostic tool for computing whether such geometrical property holds in a particular situation. 
Note that  \eqref{cod_sh} is a local condition which  contains information on the particular theory, but also depends on the background solution and space-time point one is looking at.
Thus, it might be possible to have theories which are symmetric hyperbolic only for a reduced subset of field configurations.\\



\begin{theorem}\label{theo_2}
A theory satisfying  \eqref{cod_sh} admits a collection of symmetrizers parametrized by vectors $t^a$ such that $t^a \in C_{\mathfrak{g}_{1}} \cap C_{\mathfrak{g}_{2}} = C_{\mathfrak{g}_{2}} $. Regardless of the particular choice, the resulting propagation cone $C_{H}$ is given by the closure of the union of the effective cones, i.e: $ C_{H} = \bar{C}_{\mathfrak{g}_{1}} \cup \bar{C}_{\mathfrak{g}_{2}} = \bar{C}_{\mathfrak{g}_{1}} $ 
\end{theorem}

\vspace{0.5cm}

\noindent The starting point underlying the second proof is the construction of an explicit pair ($t_{o}^{a}, n_{a}^{o}$) satisfying the positivity condition. We then show this particular couple necessarily lies in the intersections of the cones (co-cones) engendered by the effective metrics. Then, by looking at the determinant of $T_{~~\gamma }^{\alpha }~$, we will argue how much we can extend the vector $t_{o}^{a}$ and the covector $n_{a}^{o}$ without loosing the positive character of the eigenvalues. It turns out, the limits for extending them are just the closures of the above-mentioned intersections. Finally, once one finds all possible covectors $n_a$, a simple computation yields the propagation cone $C_{H}$, which we find to be 
$C_{H}= \bar{C}_{\mathfrak{g}_{1}} \cup \bar{C}_{\mathfrak{g}_{2}} $ and, as a consequence of proposition \ref{prop:inclusion}, there follows $C_{H}= \bar{C}_{\mathfrak{g}_{1}}$.



\section{Proof of the main Theorems}

To begin with, note that $H_{\alpha \beta}(t, n)$ is linear w.r.t. $t^a$ and $n_a$. Thus, once a pair $(t_{o}^{a}, n_{a}^{o})$ satisfying $H_{\alpha \beta}(t_o , n_o ) > 0 $ is found, then (by keeping fixed one of them, say $t_{o}^a$), it will exist a neighborhood (around the other, $n_{a}^{o}$) for which $H_{\alpha \beta}$ is still positive definite. These neighborhoods necessarily define open and convex sets which characterize them as cones according to our definition.

\subsection{Non-degenerate $F_{ab}$}

We start by writing $H_{\alpha \beta}(t, n)$ in a convenient representation. Our choice is irrelevant from the conceptual perspective, but important from the operational point of view. We rely on the same frame used in section II-C and consider as a basis for $\mathbb{R}^{6}$ all anti-symmetrized pairs of the covectors $\tau_a , x_a , y_a , z_a $ (see the appendix A).
A well established algebraic result \cite{hoffman1971linear}, states that the positive definiteness of $H_{\alpha \beta}(t, n)$ is equivalent to the positivity of all eigenvalues of a matrix $T^{\alpha}_{~\beta}$, related via \eqref{T}, for any Riemannian inner product  $p_{\alpha \beta }$.

Thus, we seek for a specific pair $(t_{o}^{a}, n_{a}^{o})$, rendering all eigenvalues positive. Unfortunately, we were not able to explicitly calculate them for the most general pair. However, assuming that $t_{o}^{a}$ and $n_{a}^{o}$ lie within the plane defined by the principal null directions, calculations simplify considerably.
In the chosen frame this hypothesis is equivalent to,
\begin{equation}
 t_{o}^{a}=\left( t_{0},0,0,t_{3}\right) \quad \text{ ; }  \quad n_{a}^{o}=\left(n_{0},0,0,n_{3}\right)  \nonumber
\end{equation}
For such pairs, the six eigenvalues $\lambda_i$ were computed using Mathematica, and are given by,
\begin{eqnarray}
\lambda _{1,2} &=&\left( n_{0}-n_{3}\right) \left(
t_{0}+t_{3}\right)  \label{lambdas1} \nonumber \\
\lambda _{3,4} &=&\left( n_{0}+n_{3}\right) \left(
t_{0}-t_{3}\right)  \nonumber \\
\lambda _{5} &=&\left( n_{0}t_{0}+n_{3}t_{3}\right)
\alpha _{1}\alpha _{2} \nonumber \\
\lambda _{6} &=&\left( n_{0}t_{0}+n_{3}t_{3}\right)
\beta _{1}\beta _{2} \nonumber
\end{eqnarray}
Consequently, $\lambda_5$ and $\lambda_6$ will be positive if $ sign(\alpha_1 \alpha_2 ) = sign(\beta_1 \beta_2)$. Equivalently, we have
\begin{equation}
 0 < \alpha_1 \alpha_2 \beta_1 \beta_2=\Omega_1 \Omega_2 \nonumber
\end{equation}
These relations imply the following: i) $sign(\Omega_1) = sign(\Omega_2)$ which, as we have seen, means that the two cones have a non-empty intersection; ii) recalling that $\alpha_{1}\beta_{2} \leq 1 \leq \alpha_{2}\beta_{1}$ (from lemma 1 in section II-C), one has $\alpha_1 \beta_2 >0$, which is precisely expression \eqref{cod_sh} from theorem 1. Under these assumptions, there are only two possible cases to analyze: $\Omega_i > 0 $, which corresponds to figure \ref{cones:a} (section II-C); 
$\Omega_i < 0 $, as illustrated in figure \ref{cones:c}.
These two cases are considered separately below,
\begin{enumerate}
 \item If $\Omega_i >0 $, we choose $t_{o}^{a}=\left(1,0,0,0\right) $ and $n_{a}^{o}=\left( 1,0,0,0\right) $. And we get, 
\begin{equation}
 \lambda_{1,2,3,4}=1 \quad \text{;} \quad  \lambda _{5} =\alpha_{1}\alpha_{2} > 0 \quad \text{;} \quad \lambda _{6} =\beta_{1}\beta_{2} > 0\nonumber
\end{equation}
 \item If $\Omega_i <0 $, we choose $t_{o}^{a}=\left(0,0,0,1\right) $ and $n_{a}^{o}=\left( 0,0,0,1\right) $. 
Obtaining, (once the symmetrizer is multiplied by $(-1)$),
\begin{equation}
 \lambda_{1,2,3,4}=1 \quad \text{;} \quad  \lambda _{5} = -\alpha_{1}\alpha_{2} > 0 \quad \text{;} \quad \lambda _{6} = -\beta_{1}\beta_{2} > 0\nonumber
\end{equation}
\end{enumerate}

%
Thus, we conclude that \eqref{cod_sh} is a \textit{sufficient condition}
for the system to be symmetric hyperbolic. To prove it is also \textit{necessary}, we still need to justify the restriction of $(t_{o}^{a}, n_{a}^{o})$ to those lying on the PND's plane. We will return briefly to complete this part of the proof. For the time being, we now concentrate on the question regarding how far the neighborhoods of $(t_{o}^{a}$, and $n_{a}^{o})$ can be extended, while still preserving the positivity condition. We already discard the space-like part (with respect to the effective metrics) of the null plane, for in this case some eigenvalues become negative.

The determinant of $T^{\alpha}_{~\beta}$ was calculated using Mathematica and is given by,
\begin{eqnarray} \label{det}
\det \left( T_{~\beta }^{\alpha }\right)
&=&\Omega _{1}\Omega _{2}\left( \mathfrak{g}_{1}^{ab}n_{a}n_{b}\right) \left(\mathfrak{g}_{2}^{ab}n_{a}n_{b}\right) \left( n_{a}t^{a}\right) ^{2}
\left(\mathfrak{g}^{1}_{ab}t^{a}t^{b}\right) \left( \mathfrak{g}^{2}_{ab}t^{a}t^{b}\right)  \label{det_Tnt}\\
&=& \lambda _{1}\lambda _{2}...\lambda _{6} \nonumber
\end{eqnarray}
for generic vector $t^a$ and covector $n_a$. This result is at the base of the following discussions.

First, notice that the vector $t_{o}^{a}$ (and covector $n_{a}^{o}$) we find above, belongs to the intersection of the two cones (co-cones). In other words, $t_{o}^{a} \in C_{\mathfrak{g}_{1}} \cap C_{\mathfrak{g}_{2}}$ and $n_{a}^{o} \in C^{*}_{\mathfrak{g}_{1}} \cap C^{*}_{\mathfrak{g}_{2}} $
\footnote{In our frames, this can be seen in a rather direct way for each of the two situations, namely, $\Omega_i >0 $ and $\Omega_i < 0$.
That is, why we believe it is not necessary to further justify this statement.}.
Now, whenever an eigenvalue becomes zero, one of the following must hold: (a) $n_a$ is a null covector of either $\mathfrak{g}_{1}^{ab}$ or  $\mathfrak{g}_{2}^{ab}$;
(b) $t^a$ is a null vector of either $\mathfrak{g}^{1}_{ab}$ or $\mathfrak{g}^{2}_{ab}$; (c) $t^a n_a = 0$.

\begin{proposition}
If $t^a \in C_{\mathfrak{g}_{1}} \cap C_{\mathfrak{g}_{2}} $ and $n_a \in C^{*}_{\mathfrak{g}_{1}} \cap C^{*}_{\mathfrak{g}_{2}} $, then  $t^a n_a > 0 $. \label{prop:tn}
\end{proposition}
\begin{proof}
We rely on previous results from section II-C. In particular, proposition \ref{prop:inclusion}, which states: 
 (i) $C^{*}_{\mathfrak{g}_{1}} \subseteq C^{*}_{\mathfrak{g}_{2}}  $ and (ii) $ C_{\mathfrak{g}_{2}} \subseteq C_{\mathfrak{g}_{1}}$. Thus, 
 \begin{eqnarray}
  C_{\mathfrak{g}_{1}} \cap C_{\mathfrak{g}_{2}} &\equiv& C_{\mathfrak{g}_{2}} \nonumber\\
  C^{*}_{\mathfrak{g}_{1}} \cap C^{*}_{\mathfrak{g}_{2}}   &\equiv& C^{*}_{\mathfrak{g}_{1}}  \subseteq C^{*}_{\mathfrak{g}_{2}} \nonumber
\end{eqnarray}
It is straightforward to see now that any covector 
$n_a \in C^{*}_{\mathfrak{g}_{1}} \subseteq C^{*}_{\mathfrak{g}_{2}}$ 
will satisfy (recalling the definition of dual cone) that $t^a n_a > 0 $, for any vector 
$t^a \in C_{\mathfrak{g}_{2}} $.
\end{proof}
From this result we conclude that the set of all the $t^a$ for which the symmetrizer is positive is just $  C_{\mathfrak{g}_{1}} \cap C_{\mathfrak{g}_{2}} \equiv C_{\mathfrak{g}_{2}} $; 
and similarly, that the set of all $n_a$ is given by $ C^{*}_{\mathfrak{g}_{1}} \cap C^{*}_{\mathfrak{g}_{2}} \equiv C^{*}_{\mathfrak{g}_{1}}$. As a corollary, and according to Geroch's definition of a \textit{propagation cone} (we gave in section III, def. 2), we find:
\begin{equation}
 C_{H} = \bar{C}_{\mathfrak{g}_{1}} \equiv \bar{C}_{\mathfrak{g}_{1}} \cup \bar{C}_{\mathfrak{g}_{2}}  
\end{equation}

\noindent This concludes the proof of theorem 2 for a non-degenerate $F_{ab}$. To complete the demonstration of THM 1 we proceed by contradiction. \\

To show \eqref{cod_sh} is also a \textit{necessary condition} for (symmetric) hyperbolicity, we assume $C_{\mathfrak{g}_{1}} \cap C_{\mathfrak{g}_{2}} = \emptyset$
and that there exists a pair $(t_{o}^{a}, n_{a}^{o})$ satisfying: 
$\lambda_i > 0 \text{  ,  }\forall i=1,2,..,6 $. 
This will lead us to a contradiction. Indeed, if $t_{o}^a$ (or $n_{a}^{o}$) lies outside both cones (co-cones), then the resulting set of allowed vectors (covectors) acoording to \eqref{det_Tnt} does not configure a convex cone. In fact, Eq. \eqref{det_Tnt} directly implies that, if such a vector (covector) existed one would be able to keep the positivity of the symmetrizer by reflecting the corresponding vector through the  preferred plane. However, it turns out that their sum would fall into a part of the null plane in which the system is not hyperbolic. Thus, the set of vectors wouldn't characterize a cone according to our definition, which is a contradiction. Therefore, both $t_{o}^a$ and $n_{a}^{o}$ must belong to one of their respective cones/co-cones. The contradiction came from assuming one could find a pair $(t_{o}^{a}, n_{a}^{o})$ satisfying: $\lambda_i > 0 \text{  ,  }\forall i=1,2,..,6 $ for the cases in which the cones do not intersect. Thus, \eqref{cod_sh} is also a \textit{necessary condition} which finishes the proof of THM 1 (for the non-degenerate cases).


\subsection{Degenerate $F_{ab}$}

In order to complete our proofs, we consider here the cases where the background field is degenerate. It turns out that these cases are much simpler, since (as we have already seen) there is always a non-empty intersection of the two cones. 
Thus, it only remains to be proven that there exist a common timelike vector $t_{o}^a$ and covector $n_{a}^{o}$ (respect to both effective metrics)
for which the six eigenvalues $\lambda_i$ of $T^{\alpha}_{~\beta}$ are positive. We can do this explicitly, using the results from section II-C. \\ 
Recall that a common time-like vector and covector were found, which in our particular frame reads
\begin{equation}
 t_{o}^{a}=\left(1,0,0,0\right) \quad \text{ ; } \quad n_{a}^{o}=\left( 1,0,0,0\right)  \nonumber
\end{equation}
Computing the eigenvalues for this pair one obtains,
\begin{eqnarray*}
\lambda_{1}  &=&1+\varepsilon ^{4}\left\vert b_{1}\right\vert > 0 \\
\lambda_{2} &=&1-\varepsilon ^{4}\left\vert b_{1}\right\vert > 0 \\
\lambda_{3} &=&1+\varepsilon ^{4}\left\vert b_{2}\right\vert > 0 \\
\lambda_{4} &=&1-\varepsilon ^{4}\left\vert b_{2}\right\vert > 0 \\
\lambda_{5} &=&\lambda_{6} = 1 
\end{eqnarray*}%
Then, for the second part of the proof, we follow the same lines as for the non-degenerate case. 
Notice that expressions \eqref{det_Tnt} and proposition \ref{prop:tn} also apply in the degenerate case. In addition, we have provided (in section II-C) with the analog of proposition \ref{prop:inclusion}, namely, proposition \ref{prop:deg}. Therefore, everything follows identically as before and as such concludes the proofs of theorems 1 and 2.


\section{Constraints}

The symmetrizer $h^{A}_{\phantom a\alpha}$ may be understood as map from the space of equations (indexed by $``A"$) to the space of unknowns (indexed by $``\alpha"$). In other words, it selects from the entire set of first order equations some combinations of equations which we can evolve along some direction which we usually relate to time. What are the remaining equations the symmetrizer does not capture? For the system to be consistent they should not be of the evolution type, for in that case they are either linear combination of the ones already selected by the symmetrizer or they would be incompatible with the previously chosen evolution. In other words, they must be satisfied automatically once they are satisfied initially i.e. they should be what we normally call the constraints. In Geroch's formalism a constraint is a tensor $c^{An}$ such that
\begin{equation}\label{const}
c^{A(n}K_{A \phantom a \alpha}^{\phantom a m)}=0
\end{equation}
When this formalism is applied to the equations of nonlinear electrodynamics we obtain a linear space of constraints characterized by vectors $(x,y)$ in $\mathbb{R}^{2}$ of the form $c^{An}=(xg^{an},\ yg^{an})$. 
To check that this $c^{An}$ does indeed satisfy $(\ref{const})$, we combine it with the principal symbol to obtain
\begin{equation}
c^{An}K_{A \phantom a \alpha}^{\phantom a m}=-x\left\{\frac{1}{2}g^{nmbc}+F^{nm}(\xi_{1}F^{bc}+\xi_{2}\stackrel{\ast}{F^{bc}})+\stackrel{\ast} {F^{nm}}(\xi_{2}F^{bc}+\xi_{3}\stackrel{\ast}{F^{bc}})\right\}+\frac{1}{2}y\eta^{nmbc},
\end{equation}
which is anti-symmetric in the quantities $n$ and $m$. The constraints are complete in the sense that the dimension of evolution equations provided by the symmetrizer plus the dimension of constraint equations give the correct number of PDEs. Note that, contrarily to $K_{A \phantom a \alpha}^{\phantom a m}$, the tensor $c^{An}$ does not depend on the electromagnetic field $F_{ab}$ and coincides with those of the linear theory. 
We now show that this two-dimensional system is integrable i.e. that it satisfies
\begin{equation}
\nabla_{n} (c^{An}K_{A \phantom a \alpha}^{\phantom a m}\nabla_{m}\Phi^{\alpha})=\nabla_{n}(c^{An}K_{A \phantom a \alpha}^{\phantom a m})\nabla_{m}\Phi^{\alpha}+ \nabla_{n}\nabla_{m}(c^{An}K_{A \phantom a \alpha}^{\phantom a m}\Phi^{\alpha})=0.
\end{equation}
identically. 
The second-derivative term drops out, as a consequence of \eqref{const}, and so we are
left with an algebraic equation in the first derivatives, of the form.
\begin{equation}\label{eikon}
 \frac{\partial}{\partial\Phi^{\beta}} \left(c^{An}K_{A \phantom a \alpha}^{\phantom a m}\right)\nabla_{m}\Phi^{\alpha}\nabla_{n}\Phi^{\beta}=0.
\end{equation}
This is an integrability condition for the equations as a whole. If it holds as a trivial algebraic consequence of the equations of motion we say that our constraint is integrable. To show that this is indeed the case for any NLE derived from a Lagrangian it is convenient to rewrite the equations of motion in the form
\begin{equation}\label{neweq}
\nabla_{m}(\tilde{K}_{A \phantom a \alpha}^{\phantom a m}\Phi^{\alpha})=0
\end{equation}
with $\tilde{K}_{A \phantom a \alpha}^{\phantom a m}\ \rightarrow\ \frac{1}{2}\left(\mathcal{L}_{F}\ g_{a}^{\phantom a mbc}+\mathcal{L}_{G}\ \eta_{a}^{\phantom a mbc},\ \eta_{a}^{\phantom a mbc}\ \right)$. 
Multiplying by $c^{An}\nabla_{n}$ and noting that $c^{An}$ commutes with the derivatives we obtain for the left-hand side
\begin{equation}
\nabla_{n}\nabla_{m}(c^{An}\tilde{K}_{A \phantom a \alpha}^{\phantom a m}\Phi^{\alpha}).
\end{equation}
Now, because $c^{An}\tilde{K}_{A \phantom a \alpha}^{\phantom a m}$ is also anti-symmetric in $n$ and $m$, this quantity is identically zero. 
Thus, equation \eqref{eikon} holds trivially and the constraints are therefore integrable, i.e. they remain true on the whole domain of dependence provided they were so at initial time. When the constraints are integrable, by studying the properties of the compound hyperbolic PDE system it follows that all the well-posedness results for symmetric hyperbolic systems also apply in the presence of constraints.



\section{Examples}
In this section we discuss several examples illustrating the power of the results discussed so far. The latter serve as simple diagnostic tools in testing whether well-posedness (hyperbolicity) holds and what are the associated physical propagation cones. In particular, if $\alpha _{1}\beta _{2}>0$, theorem \eqref{theo_1} guarantees a non-trivial intersection between the effective cones and so, hyperbolicity. When this condition is fulfilled we shall check whether these cones intersect the background cone or not and whether the propagation speeds (given by $\bar C_{\mathfrak{g}_{1}}$) are sub or super-luminal.
 
We will perform this analysis based on $\gamma_{1}$ for non-degenerate cases and $b_{1}$ for degenerate ones. As can be seen from propositions \eqref{prop_1} and  \eqref{prop:inclusion},  when  $\gamma_{1} < 1$, then $\bar C_{\eta} \ (or \ \bar C_{\sigma}) \subseteq \bar C_{\mathfrak{g}_{1}}$
; when $\gamma_{1}=1$, then $ \bar C_{\mathfrak{g}_{1}} = \bar C_{\eta} \ (or \ \bar C_{\sigma})$ the effective metric becomes in the background (of a rotation of it)
; finally in the sub-luminal case $ 1 < \gamma_{1}$  then $ \bar C_{\mathfrak{g}_{1}} \subseteq \bar C_{\eta} \ or \ (\bar C_{\sigma})$. 
Notice that since all the metrics share at least one null direction along that direction all propagation speeds coincide, so we always have some directions with speed of light propagation. 

For degenerate cases, the causal possibilities are giving by $b_{1}$, as it can be seen from  proposition \eqref{prop:deg} if $b_{1} < 0$ then $\bar C_{\mathfrak{g}_{1}} \subseteq \bar C_{\eta}$; if $b_{1} = 0$ then $\bar C_{\eta} = \bar C_{\mathfrak{g}_{1}}$; if $0 < b_{1}$ then $\bar C_{\eta} \subseteq \bar C_{\mathfrak{g}_{1}}$

In the next examples we shall use a frame in which $\vec{E}\parallel \vec{H}$ \eqref{coord_sistem1}, in that case, 
$\frac{\sqrt{F^{2}+G^{2}}-F}{4}=E^{2}$ and $\frac{\sqrt{F^{2}+G^{2}%
}+F}{4}=H^{2}$.

\subsection{Born-Infeld}

The Born-Infeld theory is the paramount example of non-linear electrodynamics. It was proposed to remove the divergence of the electron`s self energy at the classical level. The idea was to use a non-linear
generalization to Maxwell`s theory, which deviates from it at very strong fields. It naturally introduces a cut off $\beta $ limiting of maximum electric fields around a static charge, thus avoiding the singularity at $r\rightarrow 0.$ 
The Lagrangian expression is,     
\[
\mathcal{L} =\beta ^{2}\left( \sqrt{1+\frac{F}{2\beta ^{2}}-\frac{G^{2}}{%
16\beta ^{4}}}+1\right) 
\]

One interesting property of this theory is the absence of birefringence,
since $N_{1}=N_{2}=0,$  \cite{Obu} the two effective metrics are identical and so are the propagation speeds of all physical modes.
In this case the cones intersection is obvious, and the systems is symmetric hyperbolic. In addition the propagation cone is contained in the light cone of the space-time metric, so propagation speeds are lower or equal to the speed of light.

\begin{itemize}
\item \bigskip Non-degenerate case~$F^{2}+G^{2}\neq 0$

Before starting the hyperbolicity analysis it is important to note that when $\beta ^{2}\rightarrow \frac{\sqrt{F^{2}+G^{2}}-F}{4}$, 
$\frac{\partial \mathcal{L} }{\partial F}\rightarrow \infty $ 
and if 
$\beta^{2}<\frac{\sqrt{F^{2}+G^{2}}-F}{4}$ 
then 
$\mathcal{L} $ became complex.
Thus, we shall restrict attention to the range  $\frac{\sqrt{F^{2}+G^{2}}-F}{4}=E^{2}<\beta ^{2}$,  which is consistent with the original
idea of a limited electric field strength in the Born-Infeld theory.

The effective metrics are conformally related, 
$g_{2}^{ab}=\left( 
\frac{-G^{2}+16\beta ^{4}+8F\beta ^{2}}{4\left( 2\beta ^{2}+F\right) ^{2}}%
\right) g_{1}^{ab}$
and it easy to check that $0<\alpha _{1}\beta _{2}=1$ (and $0<\alpha _{1}\alpha
_{2}\beta _{1}\beta _{2}=1)$ so, as we have said, the theory is symmetric hyperbolic.

In addition  $\Omega _{1}=\frac{1}{\Omega _{2}}=4\frac{\left( 2\beta ^{2}+F\right) ^{2}%
}{16\left( \beta ^{2}+H^{2}\right) \left( \beta ^{2}-E^{2}\right) }>0$ 
thus when
$E^{2}<\beta ^{2}$,  the effective metrics have cone intersections with the
background metric.\\
We compute now 
\begin{eqnarray*}
\gamma _{1}^{2} &=&\left\vert \frac{\left( 1+\frac{H^{2}}{\beta ^{2}}\right) 
}{\left( 1-\frac{E^{2}}{\beta ^{2}}\right) }\right\vert 
\end{eqnarray*}%
and conclude that $1<\gamma _{1}^{2}$
i.e. it is a sub-luminal case, only propagation speeds up light are allowed.

In the limit of weak fields, $
E^{2}, H^{2}<<\beta ^{2}$, $\gamma _{1} \to 1$, so, as expected, the theory becomes closer to Maxwell's.

\item Degenerate case~$F^{2}+G^{2}=0$

As we have established, degenerate cases are always hyperbolic, we can check the cone intersection. 
In this case, using \eqref{gab_deg},  $g_{1}^{ab}=g_{2}^{ab}$, with $b_{1,2}=-\frac{1}{\beta ^{2}}<0$ the propagation speeds are lower than light.
\end{itemize}

\subsection{Toy Model 1}

The present example, like the Born-Infeld one, has only one effective metric ($N_{1}=N_{2}=0$), so it is symmetric hyperbolic for any value of the fields, but unlike the former the effective metric cone never intersects the background metric cone (see figure  \ref{cones:c}). Its Lagrangian is given by,

\[
\mathcal{L} =\frac{F}{G} 
\]

The theory is not defined for $G=0$ since $ \xi _{2}, \xi _{3}, a, Q$, and $R$ blow up in this limit, so the degenerate case will not be discussed. 

The effective metric definitions \eqref{effective_met} we are using also go bad at $F=0$, but they can be re-scaled (see footnote  \ref{foot}) so that they become finite. 
We can use that new effective metrics in place of the former without affecting the results of this paper. 
We obtain, $\bar g_{2}^{ab}=\left( -\frac{16}{G^{2}}\right) \bar g_{1}^{ab}$ so they are conformally related, and therefore share the same cone implying the system is symmetric hyperbolic. In addition $\bar \Omega _{1}=\frac{1}{\bar \Omega _{2}}=- H^{2}E^{2} <0$, so the effective metrics cone have no
intersection with the background metrics cone. This implies that initial data must be given in a time-like hyper-surface with respect the background metric, but space-like for the effective metrics, evolution occurs in temporal directions with respect to the last one.

\subsection{Electrodynamics from Kaluza-Klein theory}

We shall analyze three examples introduced in \cite{gibbons2001born}. 
The authors start from the Kaluza Klein metrics in $5=d+1$ dimensions, and adding a Gauss-Bonet terms (that in five dimension is not a topological invariant) to the Einstein Hilbert action. They arrive to an effective electrodynamics theory in $d=3+1$
dimensions.
The resulting Lagrangian is
\begin{equation}\label{2a}
\mathcal{L} =-\frac{1}{4}F+\frac{1}{16}\gamma \left( \left( b-1\right) F^{2}-%
\frac{3}{2}G^{2}\right)   
\end{equation}
with $\gamma $ a perturbation parameter from the Lagrangian, associated with  Gauss-Bonet terms, it will be interpreted as a function of physical quantities 
$\left( e,m_{e},\hslash ,c\right) $ and $b$ a parameter associated to a term
in the action quadratic in Ricci scalar. This last parameter is chosen in
order to avoid ghost propagation, and will give rise to very different electromagnetic theories.

\bigskip The Symmetric Hyperbolic condition is
\[
\alpha _{1}\alpha _{2}\beta _{1}\beta _{2}=\frac{\left( \left( F\left(
1-4b\right) \gamma +4\right) ^{2}-\left( 5\gamma -2b\gamma \right)
^{2}\left( F^{2}+G^{2}\right) \right) }{\left( 2F\gamma -2Fb\gamma +4\right).
^{2}}>0 
\]
We shall check under which conditions this is fulfilled in the following examples.

\subsubsection{\protect Gauss-Bonnet electrodynamics}

For this theory $b=1$ and it is interpreted as the first
order string theory corrections to the general relativity \cite{zwiebach1985curvature}

The Symmetric hyperbolic condition becomes
\begin{equation}
\alpha _{1}\alpha _{2}\beta _{1}\beta _{2}=\left( 1-3H^{2}\gamma \right)
\left( 1+3\gamma E^{2}\right) >0  \label{1}
\end{equation}

Because of the $sign(\gamma)$ is not defined, and this give rise to different hyperbolicity conditions  and different effective metrics, we need to study each case in particular.

\begin{itemize}
\item If $\gamma >0$ condition \eqref{1} implies $H^{2}<\frac{1}{3\gamma }$ 

\begin{itemize}

\item Non degenerate case 

The effective metrics have cone intersections with background metric ones because
\[
\Omega _{1}=\left( 1-3H^{2}\gamma \right) \left( 1+3\gamma E^{2}\right) >0
\]
In addition, propagation faster than light are allowed, this is so even for the degenerate case, as we will check
\begin{equation}
1>\gamma _{1}=\allowbreak \frac{%
1-3H^{2}\gamma }{1+3\gamma E^{2}}>0
\end{equation}
\end{itemize}

\begin{itemize}
\item Degenerate case
\[
\begin{array}{ccc}
b_{1}=3\gamma >0 
\end{array}%
\]
\end{itemize}

\item If $\gamma <0$ the dominant energy condition is satisfied and
condition (\ref{1}) implies $E^{2}<\frac{1}{3\left\vert \gamma \right\vert }$
in concordance to Gibbons and Herdeiro.

The effective metric $\mathfrak{g}_{1}$ is conformal to the background metric, so they have the same cone

\begin{itemize}
\item Non degenerate case:  $g_{1}^{ab}=(1-3\gamma H^{2}+3\gamma E^{2})\eta^{ab}$
\item Degenerate case: $g_{1}^{ab}=\eta^{ab}$
\end{itemize}
\end{itemize}

\subsubsection{Born-Infeld to second order}

For $b=-\frac{1}{2}$ and $\gamma \propto \beta ^{2}$, equation \eqref{2a} approximate
the BI Lagrangian to second order, so we recover birefringence.
\begin{itemize}
\item Non degenerate case

The condition for hyperbolicity  is,
\[
\alpha _{1}\alpha _{2}\beta _{1}\beta _{2}=\frac{16\left( 1+3\gamma
H^{2}\right) \left( 1-3\gamma E^{2}\right) }{\left( \left( 1+3\gamma
H^{2}\right) +\left( 1-3\gamma E^{2}\right) \right) ^{2}}-3>0
\]%
Then you can prove the last expression is positive  when 
\[
\frac{1}{3}<\frac{\left( 1+3\gamma H^{2}\right) }{\left( 1-3\gamma
E^{2}\right) }<3
\]

This implies non empty cones
intersection of background with effective metrics%
\[
\Omega _{1}=\frac{\left( 3\gamma H^{2}+3\gamma E^{2}+2\right) ^{2}}{\left(
3\gamma H^{2}-3\gamma E^{2}+2\right) ^{2}}\frac{\left( 3-\frac{\left(
1+3\gamma H^{2}\right) }{\left( 1-3\gamma E^{2}\right) }\right) }{\left( 1+%
\frac{\left( 1+3\gamma H^{2}\right) }{\left( 1-3\gamma E^{2}\right) }\right) 
}>0
\]%
Moreover%
\[
\gamma _{1}=\frac{3-\frac{\left( 1+3\gamma H^{2}\right) }{\left( 1-3\gamma
E^{2}\right) }}{1+\frac{\left( 1+3\gamma H^{2}\right) }{\left( 1-3\gamma
E^{2}\right) }}
\]%
In the symmetric hyperbolic range,

 If $0<\gamma \Rightarrow 1<\frac{\left( 1+3\gamma H^{2}\right) }{%
\left( 1-3\gamma E^{2}\right) }<3\Rightarrow 0<\gamma _{1}<1$ and super-luminal propagation is allowed

If $\gamma =0\Rightarrow $  the system reduce to Maxwell theory and
propagations coincide.

If $\gamma <0\Rightarrow \frac{1}{3}<\frac{\left( 1+3\gamma H^{2}\right) }{%
\left( 1-3\gamma E^{2}\right) }<1\Rightarrow 1<\gamma _{1}$ and propagation up to speed of light are allowed

\item Degenerate case%
\[
\begin{array}{ccc}
b_{1}=3\gamma 
\end{array}%
\]

If $\gamma >0$ then $b_{1}>0$ and propagation's faster than light are
allowed, 

If $\gamma =0$ then $b_{1}=0$ and we recover Maxwell theory 

If $\gamma <0$ then $b_{1}<0$ and  propagation's up to speed of
light are allowed (propagation for the unique null direction of $F_{ab}$ will be to speed of light)
\end{itemize}

\subsubsection{Euler-Heisenberg}

For $b=\frac{1}{7}$ and $\gamma \propto \alpha $ the fine structure
constant, the Euler-Heisenberg theory  becomes the effective Lagrangian for
QED due to one loop corrections \cite{Heis} 

\begin{itemize}
\item  Non degenerate case
\[
\alpha _{1}\alpha _{2}\beta _{1}\beta _{2}=270\left( \frac{7-3\gamma E^{2}}{%
-6\left( 7+3\gamma H^{2}\right) -6\left( 7-3\gamma E^{2}\right) +63}\right)
^{2}\left( \frac{6}{5}-\frac{\left( 7+3\gamma H^{2}\right) }{\left(
7-3\gamma E^{2}\right) }\right) \left( \frac{\left( 7+3\gamma H^{2}\right) }{%
\left( 7-3\gamma E^{2}\right) }-\frac{5}{6}\right)>0 
\]

\bigskip Then the systems is symmetric hyperbolic when
\[
\frac{5}{6}<\frac{\left( 7+3\gamma H^{2}\right) }{\left( 7-3\gamma
E^{2}\right) }<\frac{6}{5} 
\]
\bigskip The analysis of the propagation velocity in this case will not be studied due to its difficulty.

\item Degenerate case%
\[
b_{1}=\frac{33}{14}\gamma +\frac{9}{14}\left\vert \gamma \right\vert 
\]

If $\gamma >0$ then $b_{1}>0$ and propagation's higher than light are allowed

If $\gamma =0$ we recovery Maxwell theory

If $\gamma <0$ then $b_{1}<0$ and  propagation up to speed of  light are
allowed.
\end{itemize}

\subsection{Euler-Heisenberg II}

We present another approach given in \cite{kleinert2013fractional}, it is a strong field approximation for the Euler-Heisenberg action, with the following Lagrangian 
\[
L=-\frac{1}{4}\kappa F\left\vert FG\right\vert ^{\frac{\delta }{2}}
\]

with $\kappa =E_{c}^{-2\delta }$, a critical field and $\delta =\frac{1}{3}\left( \frac{e^{2}}{%
4\pi \hslash c\epsilon _{0}}\right) $ for spinors QED and $\delta _{S}=\frac{%
\delta }{4}$ for scalar QED.

The degenerate case are not well define because some quantities blow up. 

For the non degenerate case, we define
\[
y_{\pm }=\frac{-\left( 6\delta ^{4}+44\delta ^{3}+92\delta ^{2}+64\delta
+16\right) \pm \sqrt{16\left( \delta +2\right) ^{3}\left( 2\delta +1\right)
^{2}\left( 6\delta ^{2}+5\delta +2\right) }}{2\left( 16\delta +28\delta
^{2}-9\delta ^{4}\right) }.
\]

Thus the system is symmetric hyperbolic when

If $\left( 16\delta +28\delta ^{2}-9\delta ^{4}\right) >0$%
\[
\frac{F^{2}}{G^{2}}<y_{-}\text{ \ \ o \ \ \ }y_{+}<\frac{F^{2}}{G^{2}}\text{ 
}
\]

If $\left( 16\delta +28\delta ^{2}-9\delta ^{4}\right) <0$%
\[
y_{-}<\frac{F^{2}}{G^{2}}<y_{+}
\]

\subsection{Toy model 2}

We consider an arbitrary Lagrangian as function only of $F$,
\[
\mathcal{L} =\mathcal{L} \left( F\right).
\]%
Then \bigskip $\xi _{2}=\xi _{3}=0$ and $b_{1}:=\left( \xi _{1}+\left\vert
\xi _{1}\right\vert \right) ,b_{2}=\left( \xi _{1}-\left\vert \xi
_{1}\right\vert \right) $ .

\begin{itemize}
\item Non degenerate case

We will see that the hyperbolicity condition is 
\[
-\frac{1}{2H^{2}}<\xi _{1}<\frac{1}{2E^{2}}
\]
\end{itemize}

\begin{enumerate}
\item \bigskip If $\xi _{1}\geq 0$%

The systems is symmetric hyperbolic if $0<\alpha _{1}\beta _{2}=1-2\xi
_{1}E^{2}$ then%
\[
0\leq \xi _{1}<\frac{1}{2E^{2}}
\]

In addition 
\[
\gamma _{1}=\frac{1-2\xi _{1}E^{2}}{1+2\xi _{1}H^{2}}
\]

If   $\xi _{1}=0$  (and $\gamma_{1}=1$) the theory behaves as Maxwell, and  the birefringence effect disappear, the propagation is given for the background cone.

If  $0<\xi _{1}<\frac{1}{2E^{2}}$ $\ $then $\gamma _{1}<1$ and propagation
faster than speed of light are allowed.

Moreover 
\[
\Omega _{1}=\left( 1-2\xi _{1}E^{2}\right) \left( 1+2\xi _{1}H^{2}\right) >0
\]
then the effective metric have cone intersection with background metric.

\item If $\xi _{1}<0$%

The systems is symmetric hyperbolic if $\ 0<\alpha _{1}\beta _{2}=1+2\xi
_{1}H^{2}$ 
\[
-\frac{1}{2H^{2}}<\xi _{1}<0,
\]
and cone propagation is given for background cone,
\[
\gamma _{1}=1 \ \ \ \Omega _{1}=1.
\]

\end{enumerate}

\begin{itemize}
\item Degenerate case

If $\xi _{1}>0$ then $b_{1}=2\xi_{1}>0$ and propagation faster than light are allowed

If $\xi _{1}<0$ then $b_{1}=0$ the background metrics is the causal cone.
\end{itemize}

The particular case 
$\mathcal{L} =e^{\frac{\psi }{2}F}$ with 
$\psi =\frac{3}{2}\gamma $ is similar to the previous Gauss-Bonnet theory. 
The hyperbolicity range coincide under the change 
$E\rightarrow H$ and $H\rightarrow E$.

\section{\protect\bigskip Conclusions}

Non-linear generalization of Maxwell's theory arise naturally in many circumstances, some of them
were discussed in the examples presented above. In most cases of interest they are generated from a variational  principle involving a Lagrangian function of the two Lorentz invariant scalars one can form from Maxwell's tensor. It is well known that all these theories have a dispersion relation (set of hyper-planes allowing solutions which are constant along them) determined by the null vectors of two effective conformal class of Lorentzian metrics, which we call ${\mathfrak{g}}_1^{ab}$ and ${\mathfrak{g}}_{2}^{ab}$. These cones share a very important property, generically they have a pair of null directions in common with the background metric, thus a preferred two plane, containing these two vectors is preferred too (in some degenerate (non-generic) cases these two directions collapse into a single one).
Being null cones of conformal Lorentzian metrics we can define their propagation cones, namely the cone of all vectors whose norms with respect to each one of these metrics are positive definite.
The shared null direction property implies that either the propagation cones are nested, one inside the other, or they are not, having no vector in common. 

A first question one might ask when one is face with any one of such generalization it is whether
they are well posed, namely whether the solutions are continuous functions of their initial data.
Without this requirement the theories are powerless, they don't have any predictive power.
Using the covariant approach to symmetric hyperbolicity introduced by Geroch
we developed simple criteria for non-linear Electrodynamics theories arising from
arbitrary Lagrangians to shield well posed system of evolution equations.

The criteria we found are of two kind, and of course equivalent among each other.
One of them is geometric in nature, it says that whenever the  propagation cones of the two conformal effective metrics have a non-empty intersection the theories are symmetric hyperbolic.
There are here two cases, in one of them the conformal metrics propagation cones have also intersection with the that of the background relativistic metric, in the other not.
As can be seen from the examples, Born-Infeld, Gauss-Bonnet, Euler-Heisenberg and toy model II have cone intersection with background metric, and only the toy model I doesn't have.
In this second case the propagation has the particularity that the propagation is along space-like directions with respect to the background metric, nevertheless, taken the theory in its own, it is perfectly causal and has an initial value formulation in which data must be given along a time-like (with respect to the background metric) hyper-surface. In addition, in both cases one of the propagation cones is always inside the other so causality is always given by the cone of the metric with the largest  propagation cone, which according to our definitions is always  ${\mathfrak{g}}^1_{ab}$. There are cases where this propagation cone strictly contains the background metric cone and so we have propagation speeds larger than "light" namely larger than those allowed at the background propagation cone. In other cases the smaller cone is strictly inside the background cone and we have relativistic causality. 

The other criteria are algebraic, the cones non-empty intersection is equivalent to the positivity of the scalar expression $\alpha_{1}\beta _{2}$ which can be explicitly computed from the Lagrangian function and its derivatives and depends only on the values of the two Lorentz invariant quantities.
Whether the biggest of the metrics propagation cones is contained or contains the background metric cone depends on whether the quantity $\gamma_1$ is  bigger  or smaller  that one for non degenerate cases, and on $b_1$ is negative or positive for degenerate cases.
The algebraic criteria are very easy to check on actual examples. 
It is important to realize that whether a theory is hyperbolic or not might depend on the field strength of it, (the values the invariant scalars take), and so along evolution a perfectly nice solution might cease to be well posed. Using the algebraic conditions one could characterize the set of theories where hyperbolicity holds for all values of the field. That set of theories should be preferred.

The proof of the above mentioned criteria involves the explicit construction of hyperbolizers, thus, when trying to evolve the equations of these theories we provide a set of evolution equations which can be safely used. The propagation cones of these symmetrizers coincide with the smaller one of the effective metrics and so we recuperate for them the same causality properties as discussed above.
Thus these hyperbolizations are optimal, they allow for all possible hyper-surfaces where initial data can be given.

We also assert that the constraints arising in the theories are in all cases integrable (in the sense of Geroch). 
In some sense this is so because the constraint structure is close to the one of Maxwell. This means that in all cases, if they are satisfied for initial data, and the system is hyperbolic, they remain satisfied along evolution as long as the system remains hyperbolic and inside the corresponding domain of dependence.

Notice that at no point we use energy conditions this is so because hyperbolicity depends on second derivative conditions of the Lagrangian, while energy conditions involve only first derivatives. To some extent energy conditions are related (through energy momentum conservation) to causality conditions. Their imposition probably would prevent those cases where the effective cones have no intersection with the space-time cone. But in those cases one probably would require a different energy condition.

$\bigskip $


\appendix

\section{Frames}

In order to facilitate calculations we work in a family of preferred frames, to built it we use the principal null directions of $F_{ab}$. To display it we use spinors \cite{penrose, pirani, wald2010general}, but this is not essential.

We know that in the spinorial form the electromagnetic field is%
\begin{equation}
F_{ab}\rightarrow \epsilon _{AB}\bar{\phi}_{A^{\prime }B^{\prime }}+\epsilon
_{A^{\prime }B^{\prime }}\phi _{AB}  \label{Fab_spinors}
\end{equation}
with $\phi _{AB}=\phi _{\left( AB\right) }$ and $\epsilon _{AB}=\left( 
\begin{array}{cc}
0 & 1 \\ 
-1 & 0%
\end{array}%
\right) $

\subsection{Non-degenerate case}

In the non-degenerate case, $\phi _{AB}$ is characterized by two different principal null directions,%
\[
\phi _{AB}=\phi ~\theta _{(A} \mu _{B)} 
\]
where $\theta _{A}\mu ^{A}=1~$ and $\phi $ is a complex normalization factor.\\
Using these spinors we can construct a (complex) null tetrad,%
\begin{equation}
\begin{array}{cccc}
k^{a}\rightarrow \theta _{A}\bar{\theta}_{A^{\prime }} \ \ & l^{a}\rightarrow
\mu _{A}\bar{\mu}_{A^{\prime }} \ \ & m^{a}\rightarrow \theta _{A}\bar{\mu}%
_{A^{\prime }} \ \ & \bar{m}^{a}\rightarrow \bar{\theta}_{A^{\prime }}\mu _{A}%
\end{array}
\label{nulltetrade1}
\end{equation}%
where $k^{a}\eta _{ab}l^{b}=1,~m^{a}\eta _{ab}\bar{m}^{a}=-1$ and all other contractions vanish. 

We observe they are eigenvectors of the electromagnetic field, namely:
\[
\begin{array}{cccc}
F_{~b}^{a}k^{b}=Re\left( \phi \right) k^{a} \ \ \ \ & F_{~b}^{a}l^{b}=-Re \left( \phi \right) l^{a} \ \ \ \ & F_{~b}^{a}m^{b}=iIm\left( \phi \right)
m^{a} \ \ \ \ & F_{~b}^{a}\bar{m}^{b}=-i Im \left( \phi \right) \bar{m}^{a}%
\end{array}%
\]
So by noticing that $F=-(\phi^{2}+\bar \phi^{2})$ and $G=i(\phi^{2}-\bar \phi^{2})$, it follows that 
\[
\begin{array}{cc}
Re\left( \phi \right) =\pm \sqrt{\frac{-F+\sqrt{F^{2}+G^{2}}}{4}} & 
\ \ Im \left( \phi ~\right) =\pm \sqrt{\frac{F+\sqrt{F^{2}+G^{2}}}{4}}%
\end{array}%
\]
Notice $ G=Re\left( \phi \right) Im \left( \phi ~\right) $, and so the signs of $Re\left( \phi \right) $ and $Im \left( \phi
~\right) $ must be properly fixed according to the sign of $G$. Finally, we build our (real) tetrad as
\begin{equation}
\begin{array}{cc}
\tau ^{a}=\frac{1}{\sqrt{2}}\left( k^{a}+l^{a}\right) &  z^{a}=\frac{1}{\sqrt{2}}\left( k^{a}-l^{a}\right) \\ 
x^{a}=\frac{1}{\sqrt{2}}\left( m^{a}+\bar{m}^{a}\right) & y^{a}=\frac{i}{\sqrt{2}}\left( m^{a}-\bar{m}^{a}\right)%
\end{array}
\label{coord_sistem1}
\end{equation}
In this frame the background metric is $\eta _{ab}=Diag\left(1,-1,-1,-1\right) $; $\tau ^{a}$ is time-like while $x^{a},y^{a},z^{a}$ are
space-like vectors. The electromagnetic field and its dual read,
\begin{eqnarray}
F_{ij} &=&\left( 
\begin{array}{cccc}
0 & -E_{1} & -E_{2} & -E_{3} \\ 
E_{1} & 0 & H_{3} & -H_{2} \\ 
E_{2} & -H_{3} & 0 & H_{1} \\ 
E_{3} & H_{2} & -H_{1} & 0%
\end{array}%
\right) =\left( 
\begin{array}{cccc}
0 & 0 & 0 & -Re\left( \phi \right) \\ 
0 & 0 & Im \left( \phi \right) & 0 \\ 
0 & - Im \left( \phi \right) & 0 & 0 \\ 
Re\left( \phi \right) & 0 & 0 & 0%
\end{array}%
\right)  \label{Fab1} \\
F_{ij}^{\ast } &=&\left( 
\begin{array}{cccc}
0 & H_{1} & H_{2} & H_{3} \\ 
-H_{1} & 0 & E_{3} & -E_{2} \\ 
-H_{2} & -E_{3} & 0 & E_{1} \\ 
-H_{3} & E_{2} & -E_{1} & 0%
\end{array}%
\right) =\left( 
\begin{array}{cccc}
0 & 0 & 0 & Im \left( \phi \right) \\ 
0 & 0 & Re\left( \phi \right) & 0 \\ 
0 & -Re\left( \phi \right) & 0 & 0 \\ 
- Im \left( \phi \right) & 0 & 0 & 0%
\end{array}%
\right)  \nonumber
\end{eqnarray}
where it follows that $\vec{E}$ and $\vec{H}$ are parallel, both lying on the $z^{a}$ direction. 
And one can immediately see how calculations will become simpler by adopting this frame.

\subsubsection{Boost freedom}

Since the extend of the two principal null directions (or corresponding spinors) is arbitrary, there is a freedom on the frame choice which does not alter the above form of $F_{ab}$. 
If one considers the transformation $\theta _{A}\rightarrow \frac{1}{\varepsilon }\theta _{A}$ and $\mu ^{A}\rightarrow \varepsilon \mu ^{A}$, 
it can be notice that $\phi _{AB}$, and therefore $F_{ab}$, remains unchanged.
Thus, if we choose $\varepsilon \,\ $real, the null tetrad \eqref{nulltetrade1} changes to 
\begin{equation}
\begin{array}{cccc}
\hat{k}^{a}\rightarrow \frac{1}{\varepsilon ^{2}}\theta _{A}\bar{\theta}%
_{A^{\prime }} & \hat{l}^{a}\rightarrow \varepsilon ^{2}\mu _{A}\bar{\mu}%
_{A^{\prime }} & m^{a}\rightarrow \theta _{A}\bar{\mu}_{A^{\prime }} & \bar{m%
}\rightarrow \bar{\theta}_{A^{\prime }}\mu _{A}%
\end{array}
\label{null_tetrad_2}
\end{equation}
and in the new frame 
\[
\begin{array}{cc}
\hat{\tau}^{a}=\frac{1}{\sqrt{2}}\left( \hat{k}^{a}+\hat{l}^{a}\right) & \hat{z}^{a}=\frac{1}{\sqrt{2}}\left( \hat{k}^{a}-\hat{l}^{a}\right) \\ 
x^{a}=\frac{1}{\sqrt{2}}\left( m^{a}+\bar{m}^{a}\right) & y^{a}=\frac{i}{\sqrt{2}}\left(m^{a}-\bar{m}^{a}\right)
\end{array}%
\]
So this freedom corresponds to a boost in the null plane defined by the two null directions.
The choice of $\varepsilon$ imaginary causes a rotation of the frame components perpendicular to the null plane form by the two null directions.
Both, the metric $\eta _{ab}~$ and Maxwell tensor components remain invariant.

\subsection{Degenerated case}

As before, the electromagnetic tensor takes the spinorial form \eqref{Fab_spinors}, but now there is just a single null direction associated to it:
\begin{equation}
\phi _{AB}=~\theta _{A}\theta _{B}
\end{equation}
To complete the frame then, we shall choose an arbitrary null direction $\mu_{A} $, such that $\theta _{A}\mu ^{A}=1$ 
and we proceed building a null tetrad like in \eqref{null_tetrad_2}:
\[
\begin{array}{cccc}
k^{a}\rightarrow \frac{1}{\varepsilon ^{2}}\theta _{A}\bar{\theta}%
_{A^{\prime }} & l^{a}\rightarrow \varepsilon ^{2}\mu _{A}\bar{\mu}%
_{A^{\prime }} & m^{a}\rightarrow \theta _{A}\bar{\mu}_{A^{\prime }} & \bar{m%
}\rightarrow \bar{\theta}_{A^{\prime }}\mu _{A}%
\end{array}%
\]
where we set $k^{a}\eta _{ab}l^{b}=1,~m^{a}\eta _{ab}\bar{m}^{a}=-1$, and in analogy to the previous case any other contraction vanish. 
With $\varepsilon$ being a real positive parameter that we can freely pick up
\footnote{Notice we strongly rely on this freedom in Section II-C.2, particularly on proving lemma \ref{free-eps}.}.

It follows that, 
\begin{equation}
 F_{ab}k^{b}=0 \quad \text{;} \quad  F_{ab}l^{b}=-\varepsilon ^{2}\left( m_{a}+\bar{m}_{a}\right) \quad \text{;} \quad 
 F_{ab}m^{b}=-\varepsilon ^{2}k_{a}  \quad \text{;} \quad F_{ab}\bar{m}^{b}=-\varepsilon^{2}k_{a} \nonumber
\end{equation}
Now we build the frames as in \eqref{coord_sistem1}, and get the following expressions for $F_{ab}$ and $F_{ab}^{\ast }$
\begin{eqnarray*}
F_{ij} &=&\left( 
\begin{array}{cccc}
0 & -E_{1} & -E_{2} & -E_{3} \\ 
E_{1} & 0 & H_{3} & -H_{2} \\ 
E_{2} & -H_{3} & 0 & H_{1} \\ 
E_{3} & H_{2} & -H_{1} & 0%
\end{array}%
\right) =\left( 
\begin{array}{cccc}
0 & -\varepsilon ^{2} & 0 & 0 \\ 
\varepsilon ^{2} & 0 & 0 & -\varepsilon ^{2} \\ 
0 & 0 & 0 & 0 \\ 
0 & \varepsilon ^{2} & 0 & 0%
\end{array}%
\right) \\
F_{ij}^{\ast } &=&\left( 
\begin{array}{cccc}
0 & H_{1} & H_{2} & H_{3} \\ 
-H_{1} & 0 & E_{3} & -E_{2} \\ 
-H_{2} & -E_{3} & 0 & E_{1} \\ 
-H_{3} & E_{2} & -E_{1} & 0%
\end{array}%
\right) =\left( 
\begin{array}{cccc}
0 & 0 & \varepsilon ^{2} & 0 \\ 
0 & 0 & 0 & 0 \\ 
-\varepsilon ^{2} & 0 & 0 & \varepsilon ^{2} \\ 
0 & 0 & -\varepsilon ^{2} & 0%
\end{array}%
\right)
\end{eqnarray*}
Thus the vectors $\vec{E}=\left( \varepsilon ^{2},0,0\right) $ and $\vec{H}=\left( 0,\varepsilon ^{2},0\right) $ are orthogonal each other (i.e: $G=0$)
and with equal norms (i.e: $F=0$), as they should.

\section{Degenerate effective metrics}

In this appendix we want to show that if one of the two effective metrics becomes degenerate (non-invertible), then the system is not symmetric hyperbolic.
Such degenerate cases will occur whenever one of the variables in \eqref{alfa_beta1} becomes zero. From equation \eqref{lemma:alg} we see that only the cases $\alpha_{1} = 0$ or $\beta_{2} = 0$ are allowed.
When $\alpha_{1}=0$, say, the metric $g_{1}^{ab}$ will no longer be invertible. 
However, we stress here that the determinant from expression \eqref{det} still factorizes into four  metrics 
(${{\mathfrak{g}}}_{1}^{ab}$, $\mathfrak{g}_{2}^{ab}$, $\tilde{\mathfrak{g}}^{1}_{ab}$, $\mathfrak{g}^{2}_{ab}$),
where now $\tilde{\mathfrak{g}}^{1}_{ab}t^{a}t^{b} = \beta_{1}\left( t_{0}^{2}-t_{3}^{2}\right)$ is obviously not the inverse of ${{\mathfrak{g}}}_{1}^{ab}$.
When $\beta_2=0$, the determinant factorizes into four metrics, 
(${{\mathfrak{g}}}_{1}^{ab}$, $\mathfrak{g}_{2}^{ab}$, ${\mathfrak{g}}^{1}_{ab}$, $\tilde{\mathfrak{g}}^{2}_{ab}$), where 
$\tilde{\mathfrak{g}}^{2}_{ab}t^{a}t^{b} = -\alpha_2 \left( t_{1}^{2}+t_{2}^{2}\right)$

\begin{figure}[!ht]
\centering
    \subfloat[$\alpha_{1}=0$ and $\alpha_{2}, \beta_1, \beta_{2} \neq 0$  \label{:dummy2}]{%
      \includegraphics[width=0.25\textwidth]{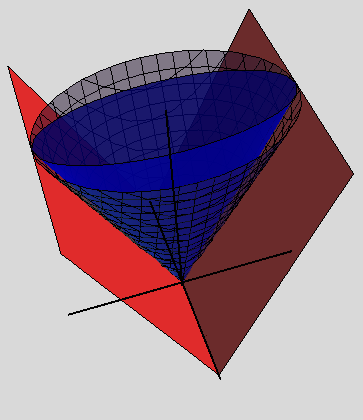}
    } \hspace{3 cm} 
    \subfloat[$\beta_{2}=0$ and $\alpha_{1}, \alpha_{2}, \beta_{1} \neq 0$ \label{subfig-2:dummy2}]{%
      \includegraphics[width=0.25\textwidth]{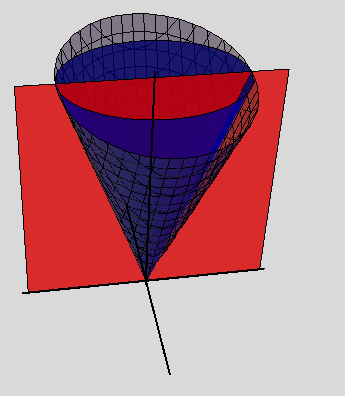}
    }
    \caption{Null surfaces of the metrics are illustrated: $\tilde{\mathfrak{g}}_{1ab}$ (red);  $ \mathfrak{g}_{2ab} $ (blue); $ \eta_{ab}$ (grey).}
    \label{fig:dummy2}
  \end{figure}

We first analyze in detail the case $\alpha_{1}=0$.
Suppose for contradiction there exist $n_{a}^{o}$ and $t_{o}^{b}$ such that the matrix $T_{~\beta}^{\alpha} \left(n^{o}, t_{o}\right) $ 
has all its eigenvalues positive, and let us generically write $n_{a}^{o}=\left( n_{0},n_{1},n_{2},n_{3}\right) $. 
Then, by following similar arguments to those used on the proofs of the main theorems, we will reach a contradiction.
The construction goes as follows: we first show it is possible to continuously connect  $n_{a}^{o}$ with a second
covector, $\hat{n}_{a}^{o}=\left( n_{0},-n_{1},-n_{2},n_{3}\right) $ without changing along the path the sign of the eigenvalues.  
The path we propose is, \\
$\left(
n_{0},n_{1},n_{2},n_{3}\right) \rightarrow \left( n_{0},\epsilon n_{1},\epsilon n_{2},n_{3}\right) \rightarrow \left( n_{0},\epsilon n_{1},0,n_{3}\right) \rightarrow \left(n_{0},\epsilon n_{1},-\epsilon n_{2},n_{3}\right) 
\rightarrow \left( n_{0},0,-\epsilon n_{2},n_{3}\right) \\
\rightarrow \left( n_{0},-\epsilon n_{1},-\epsilon n_{2},n_{3}\right)
\rightarrow \left( n_{0},-n_{1},-n_{2},n_{3}\right) $,
for some positive but small parameter $\epsilon$.\\
Then, as argued at the beginning  of section IV, since both $H_{\alpha \beta}  \left(n^{o}, t_{o}\right)$ and  $H_{\alpha \beta}  \left(\hat{n}^{o}, t_{o}\right)$ will be positive definite, it turns that $H_{\alpha \beta}  \left(n^{o} + \hat{n}^{o}, t_{o}\right)$ must be  be positive definite as well. 
But, a simple computation shows   
$n_{a}^{o}+\hat{n}_{a}^{o}=2\left(n_{0},0,0,n_{3}\right)$, 
which should also be in the cone, is null for ${\mathfrak{g}}_{1}^{ab}$, and so one of the eigenvalues must be zero,leading us to a contradiction. 
Hence, to conclude no such pair  ($n_{a}^{o}$, $t_{o}^{b}$) can exist. 

The exact same construction (but now with $t^{a}$) can be applied when  $\beta_{2} = 0$ where one of the effective metrics in their covariant version degenerates. Thus, the system is not symmetric hyperbolic for any of these two "pathological" cases. However the last case might still be strongly hyperbolic, we shall study this case in a future work.


\bibliographystyle{unsrtnat} 
\bibliography{paper2}

\begin{thebibliography}{59}
\providecommand{\natexlab}[1]{#1}
\providecommand{\url}[1]{\texttt{#1}}
\expandafter\ifx\csname urlstyle\endcsname\relax
  \providecommand{\doi}[1]{doi: #1}\else
  \providecommand{\doi}{doi: \begingroup \urlstyle{rm}\Url}\fi

\bibitem[Geroch(1996)]{Geroch}
Robert~P. Geroch.
\newblock {Partial differential equations of physics}.
\newblock 1996.

\bibitem[Euler and K\"ockel(1935)]{Eul}
H.~Euler and B.~K\"ockel.
\newblock On the scattering of light from light in the {Dirac} theory.
\newblock \emph{Naturwiss.}, 23:\penalty0 246, 1935.

\bibitem[Heisenberg and Euler(1936)]{Heis}
W.~Heisenberg and H.~Euler.
\newblock Consequences of {Dirac}s theory of positrons.
\newblock \emph{Z. Phys}, 98:\penalty0 714, 1936.

\bibitem[Karplus and Neuman(1951)]{Karp}
R.~Karplus and M.~Neuman.
\newblock The scattering of light by light.
\newblock \emph{Physical Review}, 83\penalty0 (4):\penalty0 776, 1951.

\bibitem[Schwinger(1951)]{Schwin}
J.~Schwinger.
\newblock On gauge invariance and vacuum polarization.
\newblock \emph{Physical Review}, 82\penalty0 (5):\penalty0 664, 1951.

\bibitem[Blau et~al.(1991)Blau, Visser, and Wipf]{Blau}
Steven~K. Blau, Matt Visser, and Andreas Wipf.
\newblock Analytic results for the effective action.
\newblock \emph{International Journal of Modern Physics A}, 6\penalty0
  (30):\penalty0 5409--5433, 1991.

\bibitem[Soldati and Sorbo(1998)]{Soldati}
Roberto Soldati and Lorenzo Sorbo.
\newblock Effective action for {Dirac} spinors in the presence of general
  uniform electromagnetic fields.
\newblock \emph{Physics Letters B}, 426\penalty0 (1):\penalty0 82--88, 1998.

\bibitem[Dunne(2004)]{Dunne}
Gerald~V Dunne.
\newblock {Heisenberg-Euler} effective lagrangians: basics and extensions.
\newblock \emph{arXiv preprint hep-th/0406216}, pages 445--522, 2004.

\bibitem[Shen(1984)]{Shen1984}
Yuen-Ron Shen.
\newblock Principles of nonlinear optics.
\newblock 1984.

\bibitem[Born and Wolf(1999)]{Born1999}
M.~Born and E.~Wolf.
\newblock \emph{Principles of optics: electromagnetic theory of propagation,
  interference and diffraction of light}.
\newblock Cambridge university press, 1999.

\bibitem[Bret{\'o}n(2010)]{breton2010nonlinear}
Nora Bret{\'o}n.
\newblock Nonlinear electrodynamics and cosmology.
\newblock In \emph{Journal of Physics: Conference Series}, volume 229, page
  012006. IOP Publishing, 2010.

\bibitem[De~Lorenci et~al.(2002)De~Lorenci, Klippert, Novello, and
  Salim]{de2002nonlinear}
VA~De~Lorenci, R~Klippert, M~Novello, and JM~Salim.
\newblock Nonlinear electrodynamics and frw cosmology.
\newblock \emph{Physical Review D}, 65\penalty0 (6):\penalty0 063501, 2002.

\bibitem[Novello et~al.(2004)Novello, Bergliaffa, and
  Salim]{novello2004nonlinear}
M~Novello, SE~Perez Bergliaffa, and J~Salim.
\newblock Nonlinear electrodynamics and the acceleration of the universe.
\newblock \emph{Physical Review D}, 69\penalty0 (12):\penalty0 127301, 2004.

\bibitem[Novello et~al.(2007)Novello, Goulart, Salim, and
  Bergliaffa]{novello2007cosmological}
M~Novello, E~Goulart, JM~Salim, and SE~Perez Bergliaffa.
\newblock Cosmological effects of nonlinear electrodynamics.
\newblock \emph{Classical and Quantum Gravity}, 24\penalty0 (11):\penalty0
  3021, 2007.

\bibitem[Garc{\'\i}a-Salcedo et~al.(2014)Garc{\'\i}a-Salcedo, Gonzalez, and
  Quiros]{garcia2014no}
Ricardo Garc{\'\i}a-Salcedo, Tame Gonzalez, and Israel Quiros.
\newblock No compelling cosmological models come out of magnetic universes
  which are based on nonlinear electrodynamics.
\newblock \emph{Physical Review D}, 89\penalty0 (8):\penalty0 084047, 2014.

\bibitem[Montiel et~al.(2014)Montiel, Bret{\'o}n, and
  Salzano]{montiel2014parameter}
Ariadna Montiel, Nora Bret{\'o}n, and Vincenzo Salzano.
\newblock Parameter estimation of a nonlinear magnetic universe from
  observations.
\newblock \emph{General Relativity and Gravitation}, 46\penalty0 (7):\penalty0
  1--16, 2014.

\bibitem[Pleba{\'n}ski et~al.(1984)]{plebanski1984type}
JF~Pleba{\'n}ski et~al.
\newblock Type-d solutions of the einstein and {Born-Infeld}
  nonlinear-electrodynamics equations.
\newblock \emph{Il Nuovo Cimento B Series 11}, 84\penalty0 (1):\penalty0
  65--90, 1984.

\bibitem[Bret{\'o}n(2003)]{breton2003born}
Nora Bret{\'o}n.
\newblock {Born-Infeld} black hole in the isolated horizon framework.
\newblock \emph{Physical Review D}, 67\penalty0 (12):\penalty0 124004, 2003.

\bibitem[Chemissany et~al.(2008)Chemissany, De~Roo, and
  Panda]{chemissany2008thermodynamics}
Wissam~A Chemissany, Mees De~Roo, and Sudhakar Panda.
\newblock Thermodynamics of {Born-Infeld} black holes.
\newblock \emph{Classical and Quantum Gravity}, 25\penalty0 (22):\penalty0
  225009, 2008.

\bibitem[Gunasekaran et~al.(2012)Gunasekaran, Kubiz{\v{n}}{\'a}k, and
  Mann]{gunasekaran2012extended}
Sharmila Gunasekaran, David Kubiz{\v{n}}{\'a}k, and Robert~B Mann.
\newblock Extended phase space thermodynamics for charged and rotating black
  holes and {Born-Infeld} vacuum polarization.
\newblock \emph{Journal of High Energy Physics}, 2012\penalty0 (11):\penalty0
  1--43, 2012.

\bibitem[Ayon-Beato and Garcia(1998)]{ayon1998regular}
Eloy Ayon-Beato and Alberto Garcia.
\newblock Regular black hole in general relativity coupled to nonlinear
  electrodynamics.
\newblock \emph{Physical review letters}, 80\penalty0 (23):\penalty0 5056,
  1998.

\bibitem[Bret{\'o}n and Bergliaffa(2015)]{breton2015thermodynamical}
Nora Bret{\'o}n and Santiago Esteban~Perez Bergliaffa.
\newblock On the thermodynamical stability of black holes in nonlinear
  electrodynamics.
\newblock \emph{Annals of Physics}, 354:\penalty0 440--453, 2015.

\bibitem[Born and Infeld(1933)]{Born_1}
M~Born and L~Infeld.
\newblock Foundations of the new field theory.
\newblock \emph{Nature}, 132:\penalty0 1004, 1933.

\bibitem[Born(1933)]{Born_2}
Max Born.
\newblock Modified field equations with a finite radius of the electron.
\newblock \emph{Nature}, 132:\penalty0 282, 1933.

\bibitem[Born and Infeld(1934)]{Born_3}
Max Born and Leopold Infeld.
\newblock Foundations of the new field theory.
\newblock \emph{Proceedings of the Royal Society of London. Series A,
  Containing Papers of a Mathematical and Physical Character}, pages 425--451,
  1934.

\bibitem[Gibbons(2003)]{Gib}
G.~W. Gibbons.
\newblock Aspects of {Born-Infeld} theory and string/m-theory.
\newblock \emph{Rev. Mex. Fis. 49S1}, 19, 2003.

\bibitem[Courant and Hilbert(1966)]{courant1966methods}
Richard Courant and David Hilbert.
\newblock \emph{Methods of mathematical physics}, volume~2.
\newblock CUP Archive, 1966.

\bibitem[Lax(2006)]{lax1}
Peter~D. Lax.
\newblock \emph{Hyperbolic partial differential equations}.
\newblock Courant Institute of Mathematical Sciences, 2006.

\bibitem[H{\"o}rmander(1997)]{lars1997lectures}
Lars H{\"o}rmander.
\newblock \emph{Lectures on nonlinear hyperbolic differential equations},
  volume~26.
\newblock Springer Science \& Business Media, 1997.

\bibitem[Brenier(2004)]{brenier2004hydrodynamic}
Yann Brenier.
\newblock Hydrodynamic structure of the augmented {Born-Infeld} equations.
\newblock \emph{Archive for rational mechanics and analysis}, 172\penalty0
  (1):\penalty0 65--91, 2004.

\bibitem[Lax(1973)]{lax1973hyperbolic}
Peter~D Lax.
\newblock \emph{Hyperbolic systems of conservation laws and the mathematical
  theory of shock waves}, volume~11.
\newblock SIAM, 1973.

\bibitem[Majda(2012)]{majda2012compressible}
Andrew Majda.
\newblock \emph{Compressible fluid flow and systems of conservation laws in
  several space variables}, volume~53.
\newblock Springer Science \& Business Media, 2012.

\bibitem[Serre(1999)]{serre1999systems}
Denis Serre.
\newblock \emph{Systems of Conservation Laws 1: Hyperbolicity, entropies, shock
  waves}.
\newblock Cambridge University Press, 1999.

\bibitem[Serre(2004)]{Serr}
Denis Serre.
\newblock Hyperbolicity of the nonlinear models of {Maxwell’s} equations.
\newblock \emph{Archive for rational mechanics and analysis}, 172\penalty0
  (3):\penalty0 309--331, 2004.

\bibitem[Dafermos(2010)]{dafermos2010hyperbolic}
Constantine~M Dafermos.
\newblock Hyperbolic conservation laws in continuum physics, volume 325 of
  grundlehren der mathematischen wissenschaften [fundamental principles of
  mathematical sciences], 2010.

\bibitem[Demoulini et~al.(2001)Demoulini, Stuart, and
  Tzavaras]{demoulini2001variational}
Sophia Demoulini, David~MA Stuart, and Athanasios~E Tzavaras.
\newblock A variational approximation scheme three-dimensional elastodynamics
  with polyconvex energy.
\newblock \emph{Archive for rational mechanics and analysis}, 157\penalty0
  (4):\penalty0 325--344, 2001.

\bibitem[Coleman and Dill(1971)]{coleman1971thermodynamic}
Bernard~D Coleman and Ellis~H Dill.
\newblock Thermodynamic restrictions on the constitutive equations of
  electromagnetic theory.
\newblock \emph{Zeitschrift f{\"u}r angewandte Mathematik und Physik ZAMP},
  22\penalty0 (4):\penalty0 691--702, 1971.

\bibitem[Perlick(2011)]{Volker}
Volker Perlick.
\newblock On the hyperbolicity of {Maxwell's} equations with a local
  constitutive law.
\newblock \emph{Journal of Mathematical Physics}, 52\penalty0 (4):\penalty0
  042903, 2011.

\bibitem[Speck(2012)]{speck2012nonlinear}
Jared Speck.
\newblock The nonlinear stability of the trivial solution to the
  {Maxwell-Born-Infeld} system.
\newblock \emph{Journal of Mathematical Physics}, 53\penalty0 (8):\penalty0
  083703, 2012.

\bibitem[Christodoulou(2000)]{christodoulou2000action}
Demetrios Christodoulou.
\newblock \emph{The action principle and partial differential equations}.
\newblock Number 146. Princeton University Press, 2000.

\bibitem[Christodoulou and Klainerman(1990)]{christodoulou1990asymptotic}
Demetrios Christodoulou and Sergiu Klainerman.
\newblock Asymptotic properties of linear field equations in minkowski space.
\newblock \emph{Communications on Pure and Applied Mathematics}, 43\penalty0
  (2):\penalty0 137--199, 1990.

\bibitem[Holes(2002)]{holes2002m}
Artificial~Black Holes.
\newblock M. novello, m. visser, and g. volovik.
\newblock \emph{World Scientific, Singapore}, 6:\penalty0 109, 2002.

\bibitem[Barcel{\'o} et~al.(2005)Barcel{\'o}, Liberati, Visser,
  et~al.]{barcelo2005analogue}
Carlos Barcel{\'o}, Stefano Liberati, Matt Visser, et~al.
\newblock Analogue gravity.
\newblock \emph{Living Rev. Rel}, 8\penalty0 (12):\penalty0 214, 2005.

\bibitem[Boillat(1970)]{Boi_1}
Guy Boillat.
\newblock Nonlinear electrodynamics: {Lagrangians} and equations of motion.
\newblock \emph{Journal of Mathematical Physics}, 11\penalty0 (3):\penalty0
  941--951, 1970.

\bibitem[Plebanski(1970)]{Pleb}
JF~Plebanski.
\newblock Lectures on non-linear electrodynamics.
\newblock 1970.

\bibitem[Hadamard(1903)]{hadamard1903lecons}
Jacques Hadamard.
\newblock \emph{LECONS SUR UN PROPAGATION DES ONDES}.
\newblock 1903.

\bibitem[Obukhov and Rubilar(2002)]{Obu}
Yuri~N. Obukhov and Guillermo~F. Rubilar.
\newblock Fresnel analysis of wave propagation in nonlinear electrodynamics.
\newblock \emph{Phys. Rev. D}, 66:\penalty0 024042, Jul 2002.
\newblock \doi{10.1103/PhysRevD.66.024042}.
\newblock URL \url{http://link.aps.org/doi/10.1103/PhysRevD.66.024042}.

\bibitem[Obukhov et~al.(2000)Obukhov, Fukui, and Rubilar]{obukhov2000wave}
Yuri~N Obukhov, Tetsuo Fukui, and Guillermo~F Rubilar.
\newblock Wave propagation in linear electrodynamics.
\newblock \emph{Physical Review D}, 62\penalty0 (4):\penalty0 044050, 2000.

\bibitem[Rubilar et~al.(2002)Rubilar, Obukhov, and Hehl]{rubilar2002generally}
Guillermo~F Rubilar, Yuri~N Obukhov, and Friedrich~W Hehl.
\newblock Generally covariant fresnel equation and the emergence of the light
  cone structure in linear pre-metric electrodynamics.
\newblock \emph{International Journal of Modern Physics D}, 11\penalty0
  (08):\penalty0 1227--1242, 2002.

\bibitem[Novello et~al.(2000)Novello, De~Lorenci, Salim, and
  Klippert]{novello2000geometrical}
M~Novello, VA~De~Lorenci, JM~Salim, and Renato Klippert.
\newblock Geometrical aspects of light propagation in nonlinear
  electrodynamics.
\newblock \emph{Physical Review D}, 61\penalty0 (4):\penalty0 045001, 2000.

\bibitem[De~Lorenci et~al.(2000)De~Lorenci, Klippert, Novello, and
  Salim]{de2000light}
VA~De~Lorenci, Renato Klippert, M~Novello, and JM~Salim.
\newblock Light propagation in non-linear electrodynamics.
\newblock \emph{Physics Letters B}, 482\penalty0 (1):\penalty0 134--140, 2000.

\bibitem[de~Oliveira~Costa and Bergliaffa(2009)]{de2009classification}
{\'E}rico~Goulart de~Oliveira~Costa and Santiago Esteban~Perez Bergliaffa.
\newblock A classification of the effective metric in nonlinear
  electrodynamics.
\newblock \emph{Classical and Quantum Gravity}, 26\penalty0 (13):\penalty0
  135015, 2009.

\bibitem[Penrose and Rindler(1988)]{penrose}
Roger Penrose and Wolfgang Rindler.
\newblock \emph{Spinors and space-time: Volume 2, Spinor and twistor methods in
  space-time geometry}.
\newblock Cambridge University Press, 1988.

\bibitem[Pirani(1964)]{pirani}
Felix~AE Pirani.
\newblock Introduction to gravitational radiation theory.
\newblock \emph{Lectures on general relativity}, 1:\penalty0 249--373, 1964.

\bibitem[Wald(2010)]{wald2010general}
Robert~M Wald.
\newblock \emph{General relativity}.
\newblock University of Chicago press, 2010.

\bibitem[Hoffman and Kunze(1971)]{hoffman1971linear}
K.~Hoffman and R.A. Kunze.
\newblock \emph{Linear algebra}.
\newblock Prentice-Hall mathematics series. Prentice-Hall, 1971.
\newblock URL \url{https://books.google.es/books?id=I4kQAQAAIAAJ}.

\bibitem[Gibbons and Herdeiro(2001)]{gibbons2001born}
GW~Gibbons and CAR Herdeiro.
\newblock Born-infeld theory and stringy causality.
\newblock \emph{Physical Review D}, 63\penalty0 (6):\penalty0 064006, 2001.

\bibitem[Zwiebach(1985)]{zwiebach1985curvature}
Barton Zwiebach.
\newblock Curvature squared terms and string theories.
\newblock \emph{Physics Letters B}, 156\penalty0 (5):\penalty0 315--317, 1985.

\bibitem[Kleinert et~al.(2013)Kleinert, Strobel, and
  Xue]{kleinert2013fractional}
Hagen Kleinert, Eckhard Strobel, and She-Sheng Xue.
\newblock Fractional effective action at strong electromagnetic fields.
\newblock \emph{Physical Review D}, 88\penalty0 (2):\penalty0 025049, 2013.

\end{thebibliography}

\end{document}